\documentclass[12pt]{article}

\usepackage{a4wide, amsmath,amsthm,amsfonts,amscd,amssymb,eucal,bbm,mathrsfs}

\usepackage{graphicx,tikz}
\usepackage[pdfborder={0 0 0}]{hyperref}
\usetikzlibrary{decorations.pathreplacing}
\usetikzlibrary{decorations.pathmorphing}
\usetikzlibrary{arrows}

\def\RR{{\mathbb R}}
\def\CC{{\mathbb C}}
\def\NN{{\mathbb N}}
\def\ZZ{{\mathbb Z}}

\def\HH{{\mathbb H}}
\def\SS{{\mathbb S}}

\def\A{{\mathcal A}}
\def\B{{\mathcal B}}

\def\D{{\mathcal D}}

\def\F{{\mathcal F}}
\def\H{{\mathcal H}}

\def\M{{\mathcal M}}
\def\N{{\mathcal N}}

\def\R{{\mathcal R}}
\def\S{{\mathcal S}}
\def\T{{\mathcal T}}

\def\W{{\mathcal W}}

\def\a{\alpha}
\def\b{\beta}

\def\e{\varepsilon}
\def\f{\varphi}

\def\G{\Gamma}

\def\i{\iota}
\def\k{\kappa}

\def\r{\rho}
\def\R{{\mathrm R}}

\def\sc{\Sigma}
\def\t{\tau}

\def\Ad{{\hbox{\rm Ad\,}}}

\def\dim{{\hbox{dim}\,}}

\def\sp{{\rm sp}\,}

\def\1{{\mathbbm 1}}

\def\uone{{\rm U(1)}}
\def\u1net{{\A^{(0)}}}
\def\intervals{\mathcal{I}}

\def\diff{{\rm Diff}}

\def\diffs1{\diff(S^1)}
\def\mob{{\rm M\ddot{o}b}}
\def\mob2{{\rm M\ddot{o}b}^{(2)}}

\def\supp{{\rm supp}}

\def\slim{{{\mathrm{s}\textrm{-}\lim}}}

\def\tin{{\mathrm{in}}}
\def\tout{{\mathrm{out}}}

\def\timesi{{\overset{\tin}\times}}
\def\timeso{{\overset{\tout}\times}}
\def\psl2r{{\rm PSL}(2,\RR)}
\def\sl2r{{\rm SL}(2,\RR)}
\def\su11{{\rm SU}(1,1)}
\def\2dmob{{\overline{\psl2r}\times\overline{\psl2r}}}
\def\<{\langle}
\def\>{\rangle}

\DeclareMathOperator{\tr}{tr}
\DeclareMathOperator{\sign}{sign}
\DeclareMathOperator{\Mob}{M\ddot ob}
\newcommand{\Sc}{S^1} 
\newcommand{\dd}{\,\mathrm{d}} 
\newcommand{\dE}{\,\mathrm{d}E} 
\newcommand{\ee}{\mathrm{e}} 
\newcommand{\ima}{\mathrm{i}} 
\newcommand{\cJ}{\mathcal{J}} 
\newcommand{\cI}{\mathcal{I}} 
\newcommand{\slot}{\,\cdot\,} 
\newcommand{\FerC}{ \mathrm{Fer}_\CC}

\newcommand{\Hbose}[1][]{\H_{\mathrm{\u1net}#1}}
\newcommand{\HferC}{\H_{\mathrm{Fer}_{\CC}}}

\newtheorem{theorem}{Theorem}[section]

\newtheorem{corollary}[theorem]{Corollary}
\newtheorem{proposition}[theorem]{Proposition}
\newtheorem{lemma}[theorem]{Lemma}
\theoremstyle{remark}
\newtheorem{remark}[theorem]{Remark}

\title{Construction of wedge-local nets of observables
through Longo-Witten endomorphisms. II}
\date{}
\author{{\bf Marcel Bischoff} \footnote{Supported in part by the ERC Advanced Grant 227458
OACFT ``Operator Algebras and Conformal Field Theory''.} \, and {\bf Yoh Tanimoto}$^*$\\
e-mail: {\tt bischoff@mat.uniroma2.it, tanimoto@mat.uniroma2.it}\\
Dipartimento di Matematica, Universit\`a di Roma ``Tor
Vergata''\\ Via della Ricerca Scientifica, 1 - I--00133 Roma, Italy.\\
}
\begin{document}
\maketitle
\begin{abstract}
In the first part, we have constructed several families of interacting wedge-local
nets of von Neumann algebras. In particular, there has been discovered a family of
models based on the endomorphisms of the $\uone$-current algebra
$\u1net$ of Longo-Witten.

In this second part, we further investigate endomorphisms and interacting models.
The key ingredient is the free massless fermionic net,
which contains the $\uone$-current net
as the fixed point subnet with respect to the $\uone$ gauge action. Through the
restriction to the subnet, we construct a new family of Longo-Witten endomorphisms
on $\u1net$ and accordingly interacting wedge-local nets in two-dimensional
spacetime. The $\uone$-current net admits
the structure of particle numbers and the S-matrices of the models constructed here
do mix the spaces with different particle numbers of the bosonic Fock space.
\end{abstract}

\begin{center}
{\it 
 Dedicated to Roberto Longo on the occasion of his 60th birthday
}
\end{center}

\tableofcontents

\section{Introduction}\label{introduction}
As already explained in Part I \cite{Tanimoto11-3}, construction of
interacting models of Quantum Field Theory in (physical) four spacetime dimensions has been a long-standing
open problem, and recently the algebraic approach had several progress
\cite{Lechner08, GL07, GL08, BS08, BLS11, Lechner11} and two dimensional cases work particularly well:
these works constructed models of QFT with weaker localization property,
and in some case such models turned out to be
strictly local and fully interacting \cite{Lechner08}. One should recall, however, that
the models in \cite{Lechner08} allow a complete interpretation in terms of particles
(asymptotic completeness) and the particle number is preserved under the scattering operator.
On the other hand, it is known that in four dimensions an interacting model inevitably
involves particle production \cite{Aks65}. In the present paper, we construct a further new
family of interacting wedge-local two-dimensional massless models and find that
their S-matrices mix the spaces with different particle numbers.

In fact, the requirement to involve particle production non-perturbatively is already not simple.
On the one hand, an asymptotically complete model must behave like the free theory
and hence must be compatible with the Fock space structure at asymptotic time.
On the other hand, a particle production process properly means a violation
of the Fock structure at physical time. To overcome this difficulty,
one would have to ``deform'' the free theory in a somewhat involved way
(cf. \cite{Lechner11}) or should rely on a nice trick.
Here we take the second way.
Standard examples and techniques from Conformal Field Theory provide
such a trick.

Conformal Field Theory has been well studied particularly 
on the circle, which can be seen as a chiral part of 1+1 dimensional theory.
There are many important examples of such models,
or nets in operator-algebraic terms, and both field-theoretic and
operator-algebraic techniques allow one to analyze their interrelationships.
Our trick can be briefly summarized as follows:
we consider the free complex fermionic field $\psi$ on the circle;
the field $\psi$ admits a gauge group action by $\uone$, and the fixed point
with respect to this action is known to be isomorphic to the
algebra of the $\uone$-current $J$. Both fields are free fields acting naturally
on the Fock space
(fermionic and bosonic, respectively) but the correspondence between the spaces
is quite involved. The passage to 1+1 dimensional models is simply the tensor product
of two such chiral parts.
Now, we can easily ``deform'' the two-dimensional Dirac field
(built up from the chiral parts $\psi\otimes \1$ and $\1\otimes \psi$) in such a way that it commutes with the
product action of the gauge group $\uone \times \uone$. Hence the deformation
restricts to the algebra of the conserved current $J^\mu=(J^0,J^1) = (J\otimes \1 + \1\otimes J, \1\otimes J - J\otimes\1)$, 
and this deformation is
sufficiently complicated so that the resulting S-matrix does not preserve
the bosonic Fock structure, thanks to the involved fermion-boson correspondence.

In Part I, we have constructed a family of two-dimensional massless
models based on the free current $J^\mu$ or more precisely its net $\u1net\otimes \u1net$
of von Neumann algebras of observables. The main ingredient
was endomorphisms of the algebra $\u1net(\RR_+)$ of observables localized
in the positive half-line $\RR_+$ commuting with the translations. A family of such endomorphisms
has been studied first by Longo and Witten \cite{LW11} in order to construct
Quantum Field Theory with boundary. We used those endomorphisms
to construct two-dimensional models without boundary.
In the present article, we
study the fermi net $\FerC$ generated by the free complex fermionic
field $\psi$ and its Longo-Witten endomorphisms.
We construct endomorphisms of $\FerC$ which commute with the gauge action of $\uone$,
hence restrict to the fixed point subnet of $\u1net$. It turns out that
the restricted endomorphisms cannot be implemented by second quantization operators,
hence are different from the ones considered in \cite{LW11}.
We again knit them up to construct S-matrices and wedge-local nets.

Then the fixed point  with respect to the action of $\uone\times\uone$
is considered. We find that its asymptotic behaviour is the same as
the free (bosonic) current $J^\mu$ and the S-matrix does not preserve
the space of $1+1$ particles ($1$ left and $1$ right moving particle) in the
sense of the Fock space structure.
We stress that the Fock space particle number has no intrinsic meaning as particles,
because we are in a massless case where just a scattering between two waves is considered.
One has to pass to the massive case to talk about particle production.
We will discuss in more detail the implication of this phenomenon at the end of Section \ref{wedge-local-nets}.

This paper is organized as follows.
In Section \ref{preliminaries} we recall the standard notions of
algebraic QFT and the scattering theory of two-dimensional massless
models \cite{Buchholz75, DT11-1}. Some simple observations are given
about subtheories and inner symmetries. Main examples of nets,
the free complex fermionic net $\FerC$ and the $\uone$-current net $\u1net$,
are introduced in Section \ref{examples}. Although it is well-known
\cite{KR87,Kac1998, RehrenCQFT} that the fixed point subnet $\FerC$ with respect
to $\uone$ is $\u1net$ at the field-theoretic level, we prove it
in the framework of algebraic approach. Section \ref{endomorphisms}
is devoted to the construction of new Longo-Witten endomorphisms
on $\u1net$. They are used in Section \ref{wedge-local-nets}
to construct new interacting wedge-local nets.
Outlook and open problems are summarized in Section \ref{conclusion}.

\section{Preliminaries}\label{preliminaries}

\subsection{Fermi nets on \texorpdfstring{$S^1$}{S\^{}1}}\label{fermi-nets}
Here we give a summary of one-dimensional nets, since
they will be our building blocks of the construction of two-dimensional interacting models.
In the first part, we considered {\em local} nets of von Neumann algebras on $S^1$.
Since we need to exploit the free fermionic field in this second part, a
generalized concept of nets is recalled.

We follow the definition in \cite{CKL08} and denote by $\mob2 (\cong \sl2r \cong \su11)$ the double cover of the
M\"obius group $\psl2r$. We denote by $\intervals$ \textbf{the set of 
proper intervals} $I\subset S^1$,
where proper means that $I$ is open and connected and 
neither dense nor the empty set.
A (M\"obius covariant) {\bf fermi net} is an assignment of von Neumann algebras
$\F_0(I)$ on $\H_{\F_0}$
to intervals $I\in\intervals$ on $S^1$ satisfying the following conditions:
\begin{enumerate}
\item[(1)] {\bf Isotony.} If $I_1 \subset I_2$, then $\F_0(I_1) \subset \F_0(I_2)$.
\item[(2)] {\bf M\"obius covariance.} There exists a strongly continuous unitary
representation $U_0$ of the group $\mob2$ such that
for any interval $I\in \intervals$ it holds that
\begin{equation*}
U_0(g)\F_0(I)U_0(g)^* = \F_0(gI), \mbox{ for } g \in \mob2,
\end{equation*}
where the action of $\mob2 \cong \su11$ on $S^1$ is defined through linear fractional transformation.
\item[(3)]{\bf Positivity of energy.} The generator of the one-parameter subgroup of
the lift of rotations in $\mathrm{M\ddot{o}b}$ in the representation $U_0$ is positive.
\item[(4)] {\bf Existence of the vacuum.} There is a unique (up to a phase) unit vector $\Omega_0$ in
$\H_{\F_0}$ which is invariant under the action of $U_0$,
and cyclic for $\bigvee_{I \Subset S^1} \F_0(I)$.
\item[(5)] {\bf $\ZZ_2$-grading.} There is a unitary operator $\G_0$ with $\Gamma_0^2=\1$ such that $\G_0\Omega_0 = \Omega_0$
and $\Ad \G_0(\F_0(I)) = \F_0(I)$.
\item[(6)] {\bf Graded locality.} If $I_1 \cap I_2 = \emptyset$, then $[\F_0(I_1),\Ad Z_0(\F_0(I_2))] = 0$,
where $Z_0 := \frac{\1 - \ima\G_0}{1-\ima}$.
\end{enumerate}
If the grading operator is trivial: $Z_0 = \1$, then
the net $\F_0$ is said to be {\bf local}.

Among the consequences of these conditions are (see \cite{CKL08}):
\begin{enumerate}
\item[(7)] {\bf Reeh-Schlieder property.} The vector $\Omega_0$ is cyclic and separating for each $\F_0(I)$.
\item[(8)] {\bf Additivity.} If $I = \bigcup_i I_i$, then $\F_0(I) = \bigvee_i \F_0(I_i)$.
\item[(9)] {\bf Twisted Haag duality on $S^1$.} For an interval $I\in\intervals$,
it holds that $\F_0(I)' = \Ad Z_0(\F_0(I'))$,
where $I'$ is the interior of the complement of $I$ in $S^1$.
\item[(10)] {\bf Bisognano-Wichmann property.} The modular group $\Delta_{\F_0(\RR_+)}^{\ima t}$ of $\F_0(\RR_+)$
with respect to $\Omega_0$ is equal to $U_0(\delta(-2\pi t))$, where
$S^1$ is identified as the one-point compactification of $\RR$ as below and
$\delta$ is the one-parameter group of dilations.
\item[(11)] {\bf Irreducibility.} It holds that
$\bigvee_{I\in \intervals} \F_0(I) = B(\H_{\F_0})$.
\end{enumerate}
Each algebra $\F_0(I)$ is referred to as a local algebra (even for a fermi net).
Note that if the grading operator $\G_0$ is trivial, then the definition
of fermi net coincides with the one of {\em local} M\"obius-covariant net.
We identify the circle $S^1$ and the compactified real line $\RR \cup \{\infty\}$ through
the Cayley transform
\[
t = -\ima\frac{z-1}{z+1} \Longleftrightarrow z = -\frac{t-\ima}{t+\ima}, \phantom{...}t \in \RR,
\phantom{..}z \in S^1 \subset \CC
\]
and refer to the algebra $\F_0(I)$ for an interval $I\subset \RR$.
The representation $U_0$ of $\mob2 \cong \mathrm{SL}(2,\RR)$ restricts indeed to a projective
unitary representation of $\psl2r$ \cite{CKL08}.
Let $\r$ be the ($4\pi$ periodic)  lift  
of the \textbf{rotations} in $\mathrm{PSU}(1,1)$ 
(acting by $\r(\theta)z=\ee^{\ima \theta} z$) to $\mathrm{M\ddot{o}b}^{(2)}$ and let us 
denote $R_0(\theta)=U_0(\r(\theta))=\ee^{\ima\theta L_0}$.
Under the identification between $S^1$ and $\RR \cup \{\infty\}$, one can talk about
the \textbf{translations} and \textbf{dilations} of $\RR$, which are included in $\mathrm{M\ddot{o}b}$.
In particular, the representation of translations (which we denote by $\t$)
plays a crucial role. Let us denote $T_0(t) = U_0(\t(t))$.  

A {\bf Longo-Witten endomorphism} of a fermi net $\F_0$ is an endomorphism
of the algebra $\F_0(\RR_+)$ implemented by a unitary $V_0$ which commutes with
the translation $T_0(t)$. A family of Longo-Witten endomorphisms has been found for
the $\uone$-current net and the real free fermion net \cite{LW11}.
The examples will be explained later in detail.

Note that a Longo-Witten endomorphism is uniquely implemented up to
scalar.
Indeed, since it commutes with translation, $\Ad V_0$ is an endomorphism
of $\F_0(\RR_+ + t)$ for any $t\in\RR$. If there is another unitary
$W_0$ which satisfies $\Ad W_0(x) = \Ad V_0(x)$ for any $x \in\F_0(\RR_+ + t)$,
$t \in\RR$, then by the irreducibility $W_0^*V_0$ must be scalar.

\subsection{Subnets and the character argument}\label{subnets}
Let $\F_0$ be a fermi (or local) net on $\H_{\F_0}$. Another assignment $\A_0$ of von Neumann algebras
$\{\A_0(I)\}_{I \in \intervals}$ on $\H_{\F_0}$ is called a {\bf subnet} of $\F_0$ if it satisfies
isotony, M\"obius covariance with respect to the same $U_0$ for $\F_0$
and it holds that $\A_0(I) \subset \F_0(I)$ for every interval $I \in \intervals$.
We simply write $\A_0 \subset \F_0$.
In this case, let us denote
$\H_{\A_0} = \overline{\bigvee_{I\in \intervals}\A_0(I)\Omega_0}$.
Then it is immediate to see that $\A_0(I)$ and $U_0$ restrict to $\H_{\A_0}$, and
by this restriction $\A_0|_{\H_{\A_0}}$ becomes a fermi net with the representation
of covariance $U_0|_{\H_{\A_0}}$. This restriction is also said to be a subnet
of $\F_0$ if no confusion arises.

For a fermi net $\F_0$ on $S^1$, a {\bf gauge automorphism} $\a_0$ is a family of
automorphisms $\{\a_{0,I}\}$ of local algebras which satisfies the consistency condition
\[
\a_{0,I_2}|_{\A_0(I_1)} = \a_{0,I_1} \,\,\mbox{ for } I_1 \subset I_2\,.
\]
If a gauge automorphism $\a_0$ preserves
the vacuum state $\<\Omega_0, \cdot\,\Omega_0\>$,
it is said to be an {\bf inner symmetry}. An inner symmetry $\a_0$ can be unitarily
implemented by the formula $V_{\a_0} x\Omega_0 = \a_0(x)\Omega_0$, where $x$ is an element
of some local algebra $\F_0(I)$. We say that a {\em compact} group $G$ acts on the net $\F_0$
when there is automorphisms $\{\a_{0,g}\}_{g\in G}$
which satisfy the composition law
when restricted to local algebras. The {\bf fixed point subnet} with respect to
this action of $G$ is the subnet defined by $\F_0^G(I) := \F_0(I)^G$.

Let $\F_0$ be a fermi net and $\A_0$ be a subnet.
Recall that, for a M\"obius covariant fermi net, the Bisognano-Wichmann
property is automatic. As a consequence, for each interval there
is a conditional expectation $E_{0,I}: \F_0(I) \to \A_0(I)$ which preserves
the vacuum state $\<\Omega_0, \cdot\,\Omega_0\>$ and implemented by
the projection $P_{\A_0}$ onto $\H_{\A_0}$ (see \cite[Theorem IX.4.2]{TakesakiII}).
This projection $P_{\A_0}$ contains much information of $\A_0$.

Consider the case where $\A_0 = \F_0^G$ is the fixed point subnet
with respect to an action $\a_0$ of a compact group $G$ by inner symmetry.
Then we have a unitary representation $V_{\a_0}$ of $G$ on
$\H_{\F_0}$. If we write the set of invariant vectors with respect to
$V_{\a_0}$ by $\H_{\F_0}^G$, it holds that $\H_{\F_0}^G = \H_{\A_0}$.
Indeed, the inclusion $\H_{\A_0} \subset \H_{\F_0}^G$ is obvious.
On the other hand, for $x \in \F_0(I)$, we have
\[
\left(\int_G \a_0(x) \dd g\right)\Omega_0 = \int_G \left(V_{\a_0}(g)x\Omega_0\right)\dd g,
\]
which implies that any vector in $\H_{\F_0}^G$ can be approximated from
$\H_{\A_0}$ by the Reeh-Schlieder property.

For the later use, we put here a simple observation.
\begin{proposition}\label{pr:restricting-endomorphisms}
In the situation above,
if a Longo-Witten endomorphism is implemented by $W_0$ and
$W_0$ commutes with $V_{\a_0}$, then $\Ad W_0$ restricts
to a Longo-Witten endomorphism of the fixed point subnet $\A_0$.
\end{proposition}
\begin{proof}
The unitary $W_0$ commutes with the projection $P_{\A_0}$,
hence also with the conditional expectation $E_0$ onto
$\A_0$.
\end{proof}

Let $\F_0$ be fermi (or local) net on $\H_{\F_0}$.
The Hilbert space $\H_{\F_0}$ is graded by the action 
of the rotation subgroup $R_0(\theta) = \ee^{\ima \theta L_0}$:
\[
    \H_{\F_0} 
    = \CC\Omega_0 \oplus\bigoplus_{r\in\frac 12 \NN} \H_r
    = \bigoplus_{r\in \frac 12 \NN_0} \H_r
\]
with $\H_r=\{\xi\in\H_{\F_0} : R_0(\theta)\xi = \ee^{\ima r \theta}\xi\}$
and the sum only going over $\NN_0$ for a local net. 
The {\bf conformal character} of the net $\F_0$ is 
given as a formal power series of $t=\ee^{-\beta}$:
\[
    \tr_{\H_{\F_0}}(\ee^{-\beta L_0})=\sum_{r\in\frac 12 \NN_0}^\infty \dim \H_r \cdot t^r
    \,.
\]
Let us assume that there is an action of $G=\uone$
by inner symmetry. We denote by $V_0(\theta)$ the implementing unitary. Then
$V_0$ and $U_0$ commute and $\H_{F_0}$ is graded also
by the gauge action $V_0(\theta)=\ee^{\ima \theta Q_0}$:
\[
    \H_{\F_0}
    = \CC\Omega_0 \oplus \bigoplus_{r\in\frac 12 \NN,q\in\ZZ} \H_{r,q}
    = \bigoplus_{q\in\ZZ} \H_{\slot,q},
    \qquad\text{with}\qquad \H_{\slot,q} := \bigoplus_{r\in\frac12\NN_0} \H_{r,q}
\]
and the character is given as a formal power series in $t=\ee^{-\beta}$ and $z=\ee^{-E}$:
\[ 
    \tr_{\H_{\F_0}}(\ee^{-\beta L_0 - EQ_0})=\sum_{r\in \frac12 \NN_0,q\in \ZZ}\dim \H_{r,q} \cdot t^r z^q
    \,.
\]
Recall that it holds that $\H_{\F_0}^G = \H_{\A_0}$. The operator $Q_0$
acts by $0$ on $\H_{\F_0}^G$, hence we can obtain the conformal character
of $\A_0$ just by taking the coefficient of $z^0$ in
$\tr_{\H_{\F_0}}(\ee^{-\beta L_0 - EQ_0})$.

Later in this paper we need to compare the size of two subnets.
Let $\A_0\subset\B_0\subset \F_0$ be an inclusion of three fermi nets.
If the conformal characters of $\A_0$ and $\B_0$ coincide, then this means that
the subspaces $\H_{\A_0}$ and $\H_{\B_0}$ coincide, since
we have already an inclusion $\H_{\A_0} \subset \H_{\B_0}$
and the coefficients of the conformal character are the dimensions
of eigenspaces of $L_0$.
This in turn implies that
two subnets $\A_0$ and $\B_0$ are the same since the conditional expectations
which are implemented by $P_{\A_0}, P_{\B_0}$ are the same.
We will see such an argument in an example.

\subsection{Scattering theory of waves in \texorpdfstring{$\RR^2$}{R\^{}2} (revisited)}\label{scattering-theory}
Here we just collect some basic notions regarding scattering theory of
two-dimensional massless models. As recalled in Part I \cite{Tanimoto11-3},
this theory has been established by Buchholz \cite{Buchholz75} and extended to the
wedge-local case \cite{DT11-1}.
A {\bf Borchers triple} on a Hilbert space $\H$ is a triple $(\M, T, \Omega)$ of
a von Neumann algebra $\M \subset B(\H)$, a unitary representation $T$ of $\RR^2$ on $\H$
and a vector $\Omega \in \H$ such that
\begin{itemize}
    \item $\Ad T(t_0,t_1)(\M) \subset \M$ for $(t_0,t_1) \in W_\R=\{(x_0,x_1)\in\RR^2: x_1>|x_0|\}$, the standard right wedge.
\item The joint spectrum $\sp T$ is contained in the closed forward lightcone
$\overline{V}_+ = \{(p_0,p_1) \in \RR_2: p_0 \ge |p_1|\}$.
\item $\Omega$ is a unique (up to scalar) invariant vector under $T$, and cyclic and
separating for $\M$.
\end{itemize}
We recall that one interprets $\M$ as the algebra assigned to the wedge $W_\R$.
Let $\W$ be the set of wedges, i.e. the set of all $W=gW_\R$ where $g$ is a Poincar\'e
transformation, then we define the {\bf wedge-local net} $\W\ni W\mapsto \M(W)$ associated 
with the Borchers triple $(\M,T,\Omega)$ by 
$\M(W_R+a)=T(a)\M T(a)^\ast$ and $\M(W_R'+a) = T(a)\M' T(a)^\ast$.
With the help of the modular objects one can define a representation of the Poincar\'e group
extending the one of translations $T$ \cite{Borchers92}.
For details we refer to the first part.

Take a Borchers triple $(\M,T,\Omega)$ and $x \in B(\H)$. We  write $x(a) = \Ad T(a)(x)$ for $a \in \RR^2$
and consider observables sent to lightlike directions with parameter $\T$:
\[
x_\pm(h_\T) := \int  h_\T(t) x(t,\pm t)\dd t,
\]
where $h_\T(t) = |\T|^{-\e}h(|\T|^{-\e}(t-\T))$, $0<\e<1$ is a constant,
$\T \in \RR$ and $h$ is a nonnegative symmetric smooth function on $\RR$ such that
$\int h(t)\dd t = 1$.
Then for $x \in \M$, the limits
$\Phi^\tout_+(x) := \underset{\T\to+\infty}\slim \, x_+(h_\T)$ and
$\Phi^\tin_-(x) := \underset{\T\to-\infty}\slim \, x_-(h_\T)$ exist.
Furthermore we set
$\Phi_+^\tin(y') := J_\M\Phi_+^\tout(J_\M y'J_\M)J_\M, \,\, \Phi_-^\tout(y'):=J_\M\Phi_-^\tin(J_\M y'J_\M)J_\M$
for $y' \in \M'$, where $J_\M$ is the modular conjugation of $\M$ with respect to $\Omega$.
The properties of these {\bf asymptotic fields} are summarized in \cite{DT11-1, Tanimoto11-3}.
For example, it holds that $\Phi_+^\tin(y') = \underset{\T\to-\infty}\slim \, y'_+(h_\T)$
and $\Phi_-^\tout(y') = \underset{\T\to+\infty}\slim \, y'_-(h_\T)$.

Let $\H_+$ (respectively by $\H_-$) be the space of the single excitations with positive momentum,
(respectively with negative momentum), i.e. $\H_+ = \{\xi \in \H: T(t,t)\xi = \xi \mbox{ for } t\in\RR\}$
(respectively $\H_- = \{\xi \in \H: T(t,-t)\xi = \xi \mbox{ for } t\in\RR\}$).
For $\xi_+ \in \H_+$, $\xi_- \in \H_-$, there are sequences of local operators $\{x_n\},\{y_n\} \subset \M$ and
$\{x'_n\},\{y'_n\} \subset \M'$ such that 
$\underset{n\to\infty}\slim \,P_+ x_n\Omega =\underset{n\to\infty}\slim \,P_+ x'_n\Omega  = \xi_+$ and
$\underset{n\to\infty}\slim \,P_- y_n\Omega =\underset{n\to\infty}\slim \,P_+ y'_n\Omega = \xi_-$.
We define collision states as in \cite{DT11-1}:
\[
\xi_+\timesi\xi_- = \underset{n\to\infty}\slim \,\Phi^\tin_+(x'_n)\Phi^\tin_-(y_n)\Omega, \quad
\xi_+\timeso\xi_- = \underset{n\to\infty}\slim \,\Phi^\tout_+(x_n)\Phi^\tout_-(y'_n)\Omega \,.
\]
We denote by $\H^\tin$ (respectively $\H^\tout$) the subspace generated by $\xi_+\timesi\xi_-$
(respectively $\xi_+\timeso\xi_-$).
The isometry
\[
S: \H^\tout \ni \xi_+\timeso\xi_- \longmapsto \xi_+\timesi\xi_- \in \H^\tin
\]
is called the {\bf scattering operator} or the {\bf S-matrix} of the Borchers triple $(\M,T,\Omega)$.
We say that the Borchers triple $(\M,T,\Omega)$ is {\bf interacting} if $S$ is not equal to
the identity operator on $\H^\tout$ and {\bf asymptotically complete (with respect to waves)} if
it holds that $\H^\tin=\H^\tout=\H$.

We have studied the general structure of asymptotically complete local and wedge-local nets
(using Borchers triple)
in \cite[Section 3]{Tanimoto11-3}. The point was that for a given (strictly local) $(\M,T,\Omega)$
we can construct the chiral net, and the original object $\M$ can be recovered from
the chiral net and a single operator $S$. Here we rephrase this observation from the point of view
of constructing examples based on chiral components.
See also the general structure of asymptotically complete
{\em strictly local} nets \cite[Section 3]{Tanimoto11-3}
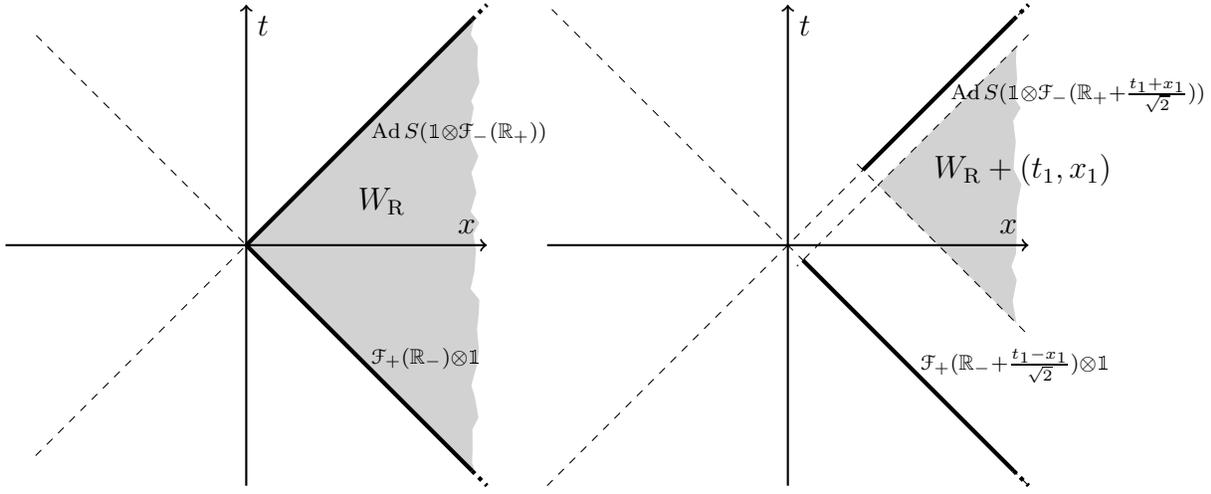
\begin{figure}[ht]
    \centering
    \definecolor{mygrey}{HTML}{d2d2d2}
    \begin{tikzpicture}[scale=0.8]
        \begin{scope}
            \fill[color=mygrey,decoration={random steps, segment length=6pt, amplitude=2pt}]
            (0,0)--(3.8,3.8) decorate{-- (3.8,-3.8)} --(0,0);
            \draw [dashed] (-3.5,-3.5)--(3.5,3.5);
            \draw [dashed] (3.5,-3.5)--(-3.5,3.5);
            \draw [thick, ->] (-4,0)--(4,0) node [above left] {$x$};
            \draw [thick, ->] (0,-4)--(0,4) node [below right] {$t$};
            \node at (2.25,.725) {$W_\mathrm{R}$};
            \draw [ultra thick] (0,0)-- node [right] {$\scriptstyle \mathrm{Ad\,} S(\1\otimes \F_-(\RR_+))$} (3.75,3.75);
            \draw [ultra thick,dotted] (3.75,3.75)--(4,4);
            \draw [ultra thick] (0,0)-- node [ right]  {$\scriptstyle\F_+(\RR_-)\otimes  \1$} (3.75,-3.75);
            \draw [ultra thick,dotted] (3.75,-3.75)--(4,-4);
        \end{scope}
        \begin{scope}[shift={(9,0)}]
            \fill[color=mygrey,decoration={random steps, segment length=6pt, amplitude=2pt}]
            (1.5,1)--(3.8,3.3) decorate{-- (3.8,-1.3)} --(1.5,1);
    	    \draw [dashed] (1.15,1.35)--(1.5,1)--(0.15,-.35);
            \draw [dashed] (4,3.5)--(1.5,1)--(4,-1.5);
            \node at (2.25,1.25) [right] {$W_\mathrm{R} +(t_1,x_1)$};
            \draw [ultra thick] (1.25,1.25)-- node [right] {$\scriptstyle \mathrm{Ad\,} S(\1\otimes \F_-(\RR_++ \frac{t_1+x_1}{\sqrt2}))$} (3.75,3.75);
            \draw [ultra thick,dotted] (3.75,3.75)--(4,4);
            \draw [ultra thick] (0.25,-0.25)--node [right] {$\scriptstyle \F_+(\RR_- +\frac{t_1-x_1}{\sqrt2})\otimes  \1$} (3.75,-3.75);
            \draw [ultra thick,dotted] (3.75,-3.75)--(4,-4);
            \draw [thick, ->] (-4,0)--(4,0) node [above left] {$x$};
            \draw [thick, ->] (0,-4)--(0,4) node [below right] {$t$};
            \draw [dashed] (-4,-4)--(4,4);
            \draw [dashed] (4,-4)--(-4,4);
        \end{scope}
    \end{tikzpicture}
    \caption{On the definition of the wedge-local net}
    \label{fig:wedge-local net}
\end{figure}
\begin{proposition}\label{pr:general-construction}
Let $\F_\pm$ be two fermi nets on $S^1$ defined on $\H_\pm$
and assume that there is a unitary operator $S$ on $\H_+\otimes\H_-$ commuting with
$T_+\otimes T_-$, leaving $\H_+\otimes \Omega_-$ and $\Omega_+\otimes\H_-$
pointwise invariant,
such that $x\otimes \1$ commutes with $\Ad S(x'\otimes \1)$ where
$x \in \F_+(\RR_-)$ and $x'\in \Ad Z_+(\F_+(\RR_+))$, and 
$\Ad S(\1\otimes y)$ commutes with $\1\otimes y'$  where
$y \in \F_-(\RR_+)$ and $y'\in \Ad Z_-(\F_-(\RR_-))$.
Then the triple
\begin{itemize}
\item $\M_S := \{x\otimes \1, \Ad S(\1\otimes y):
 x \in \F_+(\RR_-), y \in \F_-(\RR_+)\}''$,
\item $T(t,x) := T_+(\frac{t-x}{\sqrt 2})\otimes T_-(\frac{t+x}{\sqrt 2})$,
\item $\Omega := \Omega_+\otimes \Omega_-$
\end{itemize}
is an asymptotically complete Borchers triple with
the S-matrix $S$.
\end{proposition}
\begin{proof}
As in Part I \cite{Tanimoto11-3}, the conditions on $T$ and $\Omega$ are
automatic because they are just tensor products of objects for fermi nets.
Similarly, the condition that $\Ad T(a) \M_S \subset \M_S$ for $a \in W_\R$
is easily seen from the assumption that $T$ commutes with $S$ and the covariance
of fermi nets.

What remains is the cyclicity and separating property of $\Omega$ for $\M_S$.
Cyclicity is immediate because we have
\begin{eqnarray*}
\M_S\Omega &\supset& \{x\otimes \1 \cdot S(\1\otimes y)S^*\cdot \Omega: x \in \F_+(\RR_-), y \in \F_-(\RR_+) \} \\
&=& \{x\otimes y \cdot\Omega: x \in \F_+(\RR_-), y \in \F_-(\RR_+)\}
\end{eqnarray*}
by the assumed property of $S$, and the latter set is total in $\H_+\otimes \H_-$ by
the Reeh-Schlieder property for fermi nets.
As for the separating property, we define:
\[
\M_S^1 := \{\Ad S(x'\otimes \1), \1\otimes y':
 x' \in \Ad Z_+(\F_+(\RR_+)), y' \in \Ad Z_-(\F_-(\RR_-))\}''.
\]
By an analogous proof, one sees that $\Omega$ is cyclic for $\M_S^1$. Furthermore,
$\M_S$ and $\M_S^1$ commute by assumption. Hence $\Omega$ is separating for $\M_S$.
In other words, $(\M_S,T,\Omega)$ is a Borchers-triple.

It is immediate that $\Phi^\tout_+(x\otimes \1) = x\otimes \1$ and
$\Phi^\tin_-(\Ad S(\1\otimes y)) = \Ad S(\1\otimes y)$ (the latter follows since
$S$ commutes with $T$). Similarly, we have $\Phi^\tin_+(\Ad S(x'\otimes \1)) = \Ad S(x'\otimes \1)$ and
$\Phi^\tout_-(\1\otimes y') = \1\otimes y'$. From this, one concludes that
the Borchers triple $(\M_S,T,\Omega)$ is asymptotically complete and its S-matrix is $S$.
\end{proof}

We remark that we see $(\M_S,T,\Omega)$ as a fermi (i.e. twisted local) net 
defined by $\M(W_\R')=\Ad Z_+\otimes Z_- (\M)$ and that the scattering theory of waves \cite{Buchholz75} is considered to be an analogue
of the Haag-Ruelle scattering theory and it is not intended to be applied to fermionic nets.
But we will not pay much attention to this restriction, since our result
is a construction of wedge-local nets with a free massless
bosonic net as the asymptotic net, and fermionic nets appear as
auxiliary objects.

\subsection{Restriction of wedge-local nets}\label{restriction}
We consider a Borchers triple $(\M,T,\Omega)$. It is in some cases interesting to consider
a subalgebra $\N$ of $\M$. Let us denote $\H_\N := \overline{\N\Omega}$.
\begin{proposition}
If the subspace $\H_\N$ is invariant under $T$ and $\Ad T(a)(\N) \subset \N$ for $a\in W_\R$.
Then $(\N|_{\H_\N},T|_{\H_\N},\Omega)$ is a Borchers triple on $\H_\N$.
\end{proposition}
\begin{proof}
The components $\N, T$ and $\Omega$ naturally restricts to $\H_\N$.
The conditions on $T$ and $\Omega$ are trivial, even restricted to $\H_\N$.
The cyclicity of $\Omega$ is immediate from the definition of $\H_\N$.
Since $\Omega$ is already separating for $\M$, so is also for $\N$.
Endomorphic action of $T$ on $\N$ is in the hypothesis.
\end{proof}

We call a triple $(\N,T,\Omega)$ a {\bf (Borchers) subtriple} of $(\M,T,\Omega)$
if $\N$ is a subalgebra of $\M$, $\H_\N$ is invariant under $T(a)$,
$\Ad T(a)(\N) \subset \N$ for $a \in W_\R$, and $\N$ is invariant under
$\Ad \Delta_\M^{\ima t}$, where $\Delta_\M$ is the modular operator of $\M$ with respect to $\Omega$.

Recall that a Borchers triple $(\M,T,\Omega)$ gives rise to a strictly
local net if $\Omega$ is cyclic for $\M\cap \Ad T(a)(\M)'$ for any $a \in W_\R$.
We call such a triple therefore \textbf{strictly local}.
The following proposition shows that the concept of Borchers subtriple 
corresponds to the one of a local subnet.

\begin{proposition}
If a Borchers triple $(\M,T,\Omega)$ is strictly local, any subtriple $(\N,T,\Omega)$ is again strictly local when restricted on $\overline{\N\Omega}$.
\end{proposition}
\begin{proof}
Since $\N$ is invariant under the modular automorphism $\Ad \Delta_\M^{\ima t}$,
there is a conditional expectation $E$ from $\M$ onto $\N$ which preserves
the state $\<\Omega,\slot \Omega\>$ and is implemented by the projection $P_\N$
(see \cite[Theorem IX.4.2]{TakesakiII} for the original reference and
\cite[Appendix A]{Tanimoto11-2} for an application to nets).

We have to show that $\Omega$ is cyclic for the relative commutant $\N\cap \Ad T(a)(\N)'$
on the subspace $\H_\N$. Let us denote $\M_{0,a} := \M\cap\Ad T(a)(\M)'$.
We claim that $E(\M_{0,a})$ is contained in $\N\cap\Ad T(a)(\N)'$.
Indeed, by the definition of $E$, the image of $E$ is contained in $\N$.
Furthermore, if $x \in \M_{0,a}$, $y \in \Ad T(a)(\N) \subset \Ad T(a)(\M)$, then
\[
E(x)y = E(xy) = E(yx) = yE(x),
\]
hence they commute and the image $E(\M_{0,a})$ lies in the relative commutant.
Now we have
\[
\overline{\left(\N\cap\Ad T(a)(\N)'\right) \Omega} \supset \overline{E(\M_{0,a})\Omega} 
\supset \overline{P_\N\M_{0,a}\Omega} = \H_\N
\]
by the assumed strict locality of $(\M,T,\Omega)$.
\end{proof}

Let $(\B, U, \Omega)$ be an {\em asymptotically complete} local Poincar\'e covariant net on $\RR^2$
fulfilling the Bisognano-Wichmann property (see \cite{Tanimoto11-3} for related definitions).
We recall that one can define the (out-) asymptotic algebras
$\B_+\otimes\B_-$ and the scattering operator $S$
which is a unitary operator, and that it is possible to recover the original net by the formula
\[
\B(W_\R) = \{x\otimes\1, \Ad S(\1\otimes y): x\in\B_+(\RR_-), y\in\B_-(\RR_+)\}''.
\]
Note that $(\B(W_\R),U|_{\RR^2},\Omega)$ is an asymptotically complete, strictly local Borchers triple.
Here we exhibit a simple way to construct subtriples.
Let $\A_+, \A_-$ be (M\"obius covariant) subnets of $\B_+, \B_-$, respectively.
If we define
\[
\N = \{x\otimes\1, \Ad S(\1\otimes y): x\in\A_+(\RR_-), y\in\A_-(\RR_+)\}'',
\]
then $(\N,U|_{\RR^2},\Omega)$ is a Borchers subtriple of $(\B(W_\R),U|_{\RR^2},\Omega)$.
Indeed, conditions regarding $\N, U|_{\RR^2},\Omega$ are immediate. As for the invariance of
$\N$ under $\Ad \Delta_\M^{\ima t}$, it suffices to note that $S$ and $\Delta_\M^{\ima t}$ commute
(\cite[Lemma 2.4]{Tanimoto11-3}, cf.\! \cite{Buchholz75}) and that $\A_+(\RR_-)$
and $\A_-(\RR_+)$ are preserved by $\Ad \Delta_\M^{\ima t}$ because of Bisognano-Wichmann
property.

The trouble is, however, that such Borchers triples constructed as above are
not necessarily asymptotically complete in general. Indeed, the out-asymptotic
states span the subspace $\overline{\A_+(\RR_-)\Omega}\otimes\overline{\A_-(\RR_+)\Omega}$.
It is easy to see that this coincides with the full space $\overline{\N\Omega}$
if and only if it is invariant under $S$.

Since a clear-cut scattering theory is so far available only for asymptotically complete
nets, it is worthwhile to give a general condition to assure that subnets are
asymptotically complete. For simplicity, we consider the following situation:
let $\A_0$ be a fermi net on $\H_0$ with an action of a compact group $G$
by inner symmetry implemented by $V_g$.
Suppose that there is a unitary operator $S$ on $\H_0\otimes \H_0$
such that $(\M_S,T,\Omega)$ is a Borchers triple where
\begin{itemize}
\item $\M_S := \{x\otimes \1, \Ad S(\1\otimes y): x\in\A_0(\RR_-), y\in\A_0(\RR_+)\}''$,
\item $T(t,x) := T_0(\frac{t-x}{\sqrt 2})\otimes T_0(\frac{t+x}{\sqrt 2})$,
\item $\Omega := \Omega_0\otimes \Omega_0$,
\end{itemize}
as in Proposition \ref{pr:general-construction}.

\begin{proposition}\label{pr:asymptotic-completeness-fixed-point}
If $S$ commutes with $V_g\otimes V_{g'}$, $g,g'\in G$, then the triple
\begin{itemize}
\item $\N_S := \{x\otimes \1, \Ad S(\1\otimes y): x\in\A_0^G(\RR_-), y\in\A_0^G(\RR_+)\}''$ (restricted to $\overline{\N_S\Omega}$),
\item $T(t,x) := T_0(\frac{t-x}{\sqrt 2})\otimes T_0(\frac{t+x}{\sqrt 2})$,
    (restricted to $\overline{\N_S\Omega}$)
\item $\Omega := \Omega_0\otimes \Omega_0$
\end{itemize}
is an asymptotically complete Borchers triple with asymptotic algebra
$\A_0^G\otimes\A_0^G$ and  scattering operator $S|_{\overline{\N_s\Omega}}$.
\end{proposition}
\begin{proof}
As remarked above, $(\N_S,T,\Omega)$ is a Borchers triple on $\H_{\N_S}$, hence the only thing to be
proven is asymptotic completeness.
We show that the subspace $\overline{\A_0^G(\RR_-)\otimes\A_0^G(\RR_+)\Omega}$
is invariant under $S$.

We claim that $\overline{\A_0^G(\RR_-)\Omega_0}$ coincides with the subspace $\H_0^G$
of invariant vectors under $\{V_g\}_{g\in G}$. Indeed, for any $x \in \A_0$, the averaging
$\int_g V_g x\Omega_0 \dd g = \left(\int_g \a_g(x)\dd g\right) \Omega_0$ gives a projection onto
$\H_0^G$. By the Reeh-Schlieder property, any vector in $\H_0^G$ can be approximated
by vectors in $\A_0^G(\RR_-)\Omega_0$. The converse inclusion is obvious.

Now it is easy to see that $\overline{\A_0^G(\RR_-)\otimes\A_0^G(\RR_+)\Omega} =
\H_0^G\otimes \H_0^G$. This is the space of invariant vectors under
the action $\{V_g\otimes V_{g'}: g,g'\in G\}$. Since $S$ commutes with
$V_g\otimes V_{g'}$ by assumption, this subspace is preserved under $S$.
Then, as remarked before, $\overline{\N_S\Omega}$ coincides with 
$\overline{\A_0^G(\RR_-)\otimes\A_0^G(\RR_+)\Omega}$ and
we obtain the asymptotic completeness.

The statement on S-matrix is immediate from the definition
and by Proposition \ref{pr:general-construction}.
\end{proof}

\section{Examples of fermi nets}\label{examples}

\subsection{\texorpdfstring{$\uone$}{U\^{}1}-current net \texorpdfstring{$\u1net$}{A(0)}}\label{u1-current}
Let $U_1$ be the irreducible unitary positive-energy representation of $\Mob$ with lowest weight $1$
on a Hilbert space denoted by $\Hbose^1$, which can be identified with the one-particle space of the $\uone$-current.
This has the following concrete realization: consider $C^\infty(\Sc,\RR)$, where we write 
the periodic function $f\in C^\infty(\Sc,\RR)$ as a Fourier series
\begin{equation*}
    f(\theta) = \sum_{k\in\ZZ} \hat f_k \ee^{\ima k \theta}, \quad
    \hat f_k = \int_0^{2\pi} \ee^{-\ima k\theta}f(\theta) \frac{\dd \theta}{2\pi} = \overline{\hat f_{-k}} \, .
\end{equation*} 
We introduce a semi-norm
\begin{equation*}
    \| f \|^2= \sum_{k=1}^\infty k\cdot | \hat f_k|^2
\end{equation*}
and a complex structure, i.e. an isometry $\cJ$ w.r.t. $\|\slot\|$ satisfying 
$\cJ^2=-1$, by
$\cJ:\hat f_k \mapsto -\ima \sign(k) \hat f_k$ and finally we get the Hilbert space 
$\Hbose^1= \overline{C^\infty(S^1,\RR) /\RR}^{\|\slot\|}$ by completion with respect to the norm $\|\slot\|$, 
where $\RR$ is identified with the constant functions. By abuse of notation 
we denote also the image of $f\in C^\infty(S^1,\RR)$ in $\Hbose^1$ by $f$.
The scalar product (linear in the second component) and the sesquilinear form $\omega(\slot,\slot)\equiv\Im \langle\slot,\slot\rangle$ are given by 
\begin{align*}
    \langle f,g\rangle&= \sum_{k=1}^\infty k \hat f_k \hat g_{-k}\,,
    &\omega(f,g) &=\frac{-\ima}2\sum_{k\in\ZZ}k \hat f_{k} \hat g_{-k} 
    = \frac 1 2 
    \int_0^{2\pi} f(\theta)g'(\theta)\frac{\dd\theta}{2\pi} = \frac{1}{4\pi}\int f \dd g
    \,,
\end{align*}
respectively.
The unitary action $U_1$ of $\Mob$ on $\Hbose^1$ is induced by the action on $C^\infty(S^1,\RR)$
$(U_1(g)f)=(g_\ast f)(\theta):=f(g^{-1}(\theta))$.

For $I\in \cI$ we denote by $H(I)$ the closure of the subspace of real functions with 
support in $I$. This space is standard (i.e. $H(I)+\ima H(I)$ is dense in $\Hbose^1$ and
$H(I)\cap \ima H(I) = \{0\}$) and the family $\{H(I)\}_{I\in\cI}$ is a local M\"obius covariant
net of standard subspaces \cite{Longo08, LW11}. 

We explain briefly the bosonic second quantization procedure in general.
Let $\H^1$ be a separable Hilbert space, the \textbf{one-particle space}, 
and $\omega(\slot,\slot)=\Im\langle\slot,\slot\rangle$ the
sesquilinear form.
There are unitaries $W(f)$ for $f\in\H^1$ fulfilling
\[
    W(f)W(g) = \ee^{-\ima \omega(f,g)} W(f+g) = \ee^{-2\ima\omega(f,g)}W(g)W(f)
\]
and acting naturally on the bosonic Fock space $\ee^{\H^1}$ over $\H^1$.
This space is given by $\ee^{\H^1} = \oplus_{n=0}^\infty P_n(\H^1)^{\otimes n}$,
where $P_n$ is the projection 
$P_n(\xi_1\otimes\cdots\otimes \xi_n)=
1/n!\sum_{\sigma} \xi_{\sigma(1)}\otimes\cdots\otimes \xi_{\sigma(n)}$ and the 
sum goes over all permutations. 
The set of coherent vectors
$\ee^h := \oplus_{n=0}^\infty h^{\otimes n}/\sqrt{n!}$ with $h\in\H^1$ is 
total in $\ee^{\H^1}$ and it holds $\langle\ee^f, \ee^h\rangle=\ee^{\langle f,h \rangle}$.
The vacuum is given by $\Omega =\ee^0$ and the action of $W(f)$ is given by 
$W(f)\ee^0=\ee^{-\frac12\|f\|^2}\ee^{f}$, in other words the vacuum 
representation is characterized by 
$\phi(W(f))=\ee^{-\frac 12 \|f\|^2}$, where $\phi(\slot) = \langle\Omega,\slot\Omega\rangle$.
For a real subspace $H\subset \H^1$, we define the von Neumann algebra 
\begin{align*} 
    R(H) = \lbrace W(f):f\in H\rbrace''\subset B(\ee^{\H^1})
    \,.
\end{align*}
Let $U$ be a unitary on the one-particle space $\H^1$ then 
$\ee^U:=\oplus_{n=0}^\infty U^{\otimes n}$ acts on coherent states 
by $\ee^U\ee^h = \ee^{Uh}$ and is therefore a unitary on $\ee^{\H^1}$, the \textbf{second quantization unitary}.

We obtain the \textbf{$\uone$-current net} $\u1net$ on $\Hbose:=\ee^{\Hbose^1}$ with $\Omega_0=\ee^0$ 
by defining $\u1net(I):=R(H(I))$ which is covariant with respect to $U(g):=\ee^{U_1(g)}$.
For $f\in C^\infty(\Sc,\RR)$ we consider a self-adjoint operator $J(f)$ given by the generator 
of the unitary one-parameter group $W(t\cdot f)=\ee^{\ima t\cdot J(f)}$ with $t\in\RR$.
This defines the usual current (field operator)  
smeared with the real test function $f$, which fulfills $J(f)\Omega_0 = f\in\Hbose^1$ and
\begin{align*}
    [J(f),J(g)] &= 2\ima \omega (f,g) = \sum_k  k \hat f_{k}\hat g_{-k} = \frac{\ima}{2\pi}\int f \dd g 
    \,.
\end{align*}
It can be extended to complex test functions via $J(f+\ima g)=J(f)+\ima J(g)$,
and one obtains the usual operator valued ($z$-picture) distribution $J(z)$
with the relations 
\begin{align*}
    J(f) &= \sum_{n\in\ZZ} \hat f_n J_n 
    = \oint_{S^1} f(z) J(z) \frac{\dd z}{2\pi \ima },
    &J(z) = \sum_n J_n z^{-n-1}
    \\ [J_m,J_n]&=m\delta_{m+n,0}
    \,,
\end{align*}
where the modes $J_n=J(e_n)$ with $e_n(\theta)=\ee^{\ima n \theta}$ satisfy
 $J_n\Omega_0 = 0$ for $n\geq 0$.

The space $\Hbose$ is spanned by vectors of the form $\xi = J_{-n_1} \cdots J_{-n_k}\Omega_0$ 
with $0<n_1\leq \cdots\leq n_k$ with ``energy'' $N=\sum_m n_m$, 
i.e. $R(\theta)\xi = \ee^{\ima N \theta} \xi$. Therefore it is graded with respect to 
the rotations 
\begin{align*}
    \Hbose 
    &= \CC\Omega_0 \oplus\bigoplus_{n\in \NN} \Hbose[,n]&
    \Hbose[,n] &= \bigoplus_{k=1}^n \bigoplus_{\substack{0< n_1\leq\cdots\leq n_k \\ n_1+\cdots + n_k=n}} \CC 
    J_{-n_1} \cdots J_{-n_k}\Omega_0
\end{align*}
and $\dim \Hbose[,n]$ is the number of partitions of $n$ elements,
whose generating function $p(t)$ is
the inverse of Euler's function $\phi(t)=\prod_{k=1}^\infty (1-t^k)$ and therefore 
the conformal character of the $\uone$-current net is given by ($t=\ee^{-\beta}$):
$$
    \tr_{\Hbose}(\ee^{-\beta L_0})=\sum_{n=0}^\infty \dim \Hbose[,n] \cdot t^n = \prod_{n\in\NN} (1-t^n)^{-1}
$$
(a conformal character is defined as a formal power series, but
it is often convergent for $|t| < 1$ and here we used the formula
$(1-z)^{-1} = 1 + z + z^2 \cdots$).
It will be convenient to use the real parametrization $x\in\RR \cong S^1\setminus \{-1\}$
of the cut circle and use the conventions
\begin{align*}
    f(s) = \int_\RR \ee^{-\ima sp} \hat f(p)\dd p.
\end{align*}
By writing $f(s)=f_0(\theta(s))$ for $f_0\in C^\infty(S^1,\RR)$
where $\theta(s) = 2\arctan(s)$,
the space $\Hbose^1$ above can be identified with the space $L^2(\RR_+, p\dd p)$  
in which the space $\S(\RR,\RR)$ embeds by restriction of the Fourier transformation 
to $\RR_+$. In other words $\Hbose^1$ can be seen as the closure of the space 
$\S(\RR,\RR)$ with complex structure
$\cJ \hat f(p) = \ima \sign(p) \hat f(p)$ and the scalar product and sesquilinear form given by:
\begin{align*} 
    \langle f,g \rangle &= \int_{\RR_+} \hat f(-p) \hat g(p) p \dd p\,, 
    &\omega(f,g) &= \frac{-\ima}2 \int_{\RR} \hat f(-p)\hat g(p) p \dd p
    = \frac1{4\pi} \int_\RR f(x)g'(x) \dd x \,.
\end{align*}
Using the above identification we denote for $f\in\S(\RR,\RR)$
by $J(f)$ the smeared current
with $J(f)\Omega_0 = f \in\Hbose^1$. In this parametrization commutation relations read:
\begin{align*}
    [J(f),J(g)] &= \frac\ima{2\pi} \int_\RR f(x)g'(x) \dd x 
    =  \int_{\RR} \hat f(-p)\hat g(p) p \dd p\,.
\end{align*}

\subsection{The free complex fermion net \texorpdfstring{$\FerC$}{FerC}}\label{free-fermion}
We construct the net of the free complex fermion on the circle, which can be seen 
as the chiral part of the net of the free massless Dirac (or complex) fermion on two dimensional Minkowski space.
The notations of this section are basically in accordance with \cite{Wassermann98}, but we use 
a different convention of positive-energy, which leads to the conjugated complex structure.
For giving a simple description of the one-particle space, we consider first the Hilbert space $L^2(S^1)$ and the Hardy space $H^2(S^1)$, namely
\[
H^2(S^1) := \left\{f: \mbox{ analytic on the unit disk } D,
\sup_{0\le r < 1}\int_0^{2\pi} |f(r \ee^{\ima \theta})|^2 \dd \theta < \infty\right\}
\,.
\]
Any function in $H^2(S^1)$ has a $L^2$-boundary value and can be considered
as an element of $L^2(S^1)$. In this sense, $H^2(S^1)$ is a subspace of $L^2(S^1)$.
Furthermore, it holds that
\[
H^2(S^1) = \{f\in L^2(S^1): \hat f_n = 0 \mbox{ for } n < 0\},
\]
where $\hat f_n$ is the $n$-th Fourier component of $f$. We denote the orthogonal
projection onto $H^2(S^1)$ by $P$.

The group
\[
\mathrm{SU}(1,1) := \left\{\left(\begin{matrix} \a & \b \\ \overline{\b} & \overline{\a}\end{matrix}\right)
\in M_2(\CC): |\a|^2-|\b|^2 = 1\right\}
\]
acts on the circle $S^1$ by $g\cdot z = \frac{\a z + \b}{\overline{\b}z+\overline{\a}}$ and there is a unitary action of $\mathrm{SU}(1,1)$ on $L^2(S^1)$ by
\[
(U(g)f)(z) := (V_gf)(z)= \frac{1}{-\overline{\b} z+\a}f(g^{-1}\cdot z)
\,.
\]
One sees that the projection $P$ commutes with $V_g$ , since $V_gf$ is still 
an analytic function for $|\a| > |\b|$.

Then one defines a new Hilbert space $\HferC^1 = \overline{PL^2(S^1)} \oplus (\1-P)L^2(S^1)$:
namely, $\HferC^1$ is identical with $L^2(S^1)$ as a real linear space and the multiplication
by $\ima$ is given by $-\ima(2P-\1)$, or in other words, by $-\ima$ on $PL^2(S^1)$ and $\ima$ on $(\1-P)L^2(S^1)$.
Because $P$ and $U(g)$ commute, the action of $\mathrm{SU}(1,1)$ remains unitary on $\HferC^1$.

Then for $I\in\cI$ one takes real Hilbert subspaces $K(I) := L^2(I)$ of $\HferC^1$. This subspaces turn
out to be standard \cite[Theorem (p.\! 497)]{Wassermann98}. If $I_1$ and $I_2$ are disjoint intervals, $K(I_1)$ 
are real orthogonal to $K(I_2)$,
in other words $K(I_1)\subset K(I_2)^\perp$, where 
$K^\perp = \{\xi\in\H: \Re \langle \xi,K\rangle=0\}$. It turns out that 
$I\mapsto K(I)$ is a twisted-local M\"obius covariant net of standard subspaces.

We briefly explain the fermionic second quantization in general.
Let $\H^1$ be a complex Hilbert space and $\H=\Lambda (\H^1)$ the antisymmetric (fermionic) Fock space
obtained by completing the exterior algebra with the inner product.
For $A\in B(\H^1)$ with $\|A\|\leq 1$ we define $\Lambda(A)$ to 
be $A^{\otimes k}$ on $\H^k := \Lambda^k(\H^1) \subset (\H^1)^{\otimes k}$. The space is $\ZZ_2$ graded by $\Gamma:=\Lambda(-\1)$.
We define $Z=\frac{\1-\ima \Gamma}{1-\ima}$ and note that $Z^2=\Gamma$.
For $f\in \H^1$ let $a(f)$ be the bounded operator obtained by continuing the 
exterior multiplication $f\wedge \cdot$. The operators fulfill the complex Clifford relations 
$a(f)^\ast a(g)+a(g)a(f)^\ast =\langle f,g \rangle$ and
$\{a(f),a(g)\} = \{a(f)^\ast, a(g)^\ast\} = 0$ for all $f,g\in\H^1$. 
For a standard subspace $K\subset \H^1$ we define the von Neumann algebra
\[
C(K) = \{ c(f) : f\in K \}'' \subset B(\Lambda\H^1)
\]
where $c(f)=a(f)+a(f)^\ast$, which fulfills the real Clifford relations $c(f)c(g)+c(g)c(f) =2\Re\< f,g\>$.
By $\Omega=1 \in \Lambda^0$ we denote the vacuum which is cyclic and separating for $C(K)$
for every standard subspace $K\subset \H^1$. Further it holds 
\textbf{Haag-Araki duality}, i.e. $C(K^\perp)$
equals $C(K)^\perp:=Z C(K)' Z^\ast$, the 
\textbf{twisted commutant} of $C(K)$. For a unitary $U$ on $\H^1$ 
it holds $\Lambda(U)c(f)\Lambda(U^\ast) = c(Uf)$, which implies that $C$ is covariant with respect to the 
unitaries $U(\H^1)$, i.e. $\Lambda(U)C(K)\Lambda(U)^\ast = C(UK)$.

We note that in the case like the complex fermion the one-particle space is obtained from a Hilbert space $\H^1$ 
(the space of test functions) and a projection $P$ by 
$\H^1_P=P\H^1\oplus \overline{ P^\perp \H^1}$ and one gets a new representation of 
the complex Clifford algebra on $\Lambda(\H^1_P)$ by
$a_P(f)=a(Pf)+a(\overline{P^\perp f})^\ast$
where $a(f)$ is the creation operator. For a standard subspace $K\subset \H^1_P$ which is invariant
under the multiplication of $\ima_{\H^1}$ in $\H^1$,
the von Neumann algebra $C(K)$ on $\Lambda(\H_P^1)$ 
coincides with the von Neumann algebra $\{a_P(f),a_P(f)^\ast : f\in K\}''$.
Indeed, the one inclusion follows from $c(f)=a_P(f)+a_P(f)^\ast$ and the 
other follows from Araki-Haag duality and 
$\{a_P(f),c(g)\}=\langle g,f\rangle_{\H^1} = \Re \langle g,f\rangle_{\H^1_P} -\ima \Re \langle g,\ima_{\H^1} f\rangle_{\H^1_P} =0$
for $f\in K$ and $g\in K^\perp$.
We further note that the space $\Lambda(\H^1_P)$ is
as a real Hilbert space the same 
as $\Lambda(\H^1)$ and can be identified canonically with 
$\Lambda(P\H^1)\otimes \Lambda(\overline{ P^\perp \H^1})$.

We turn to the concrete case where $\H^1 = L^2(S^1)$ and
define the net $\FerC(I) := C(K(I)) = \{a_P(f),a_P(f)^\ast : f \in L^2(I)\}''$ (where here
$a_P(f) := a(\overline {Pf}) + a( P^\perp f)^\ast$) on 
$\HferC=\Lambda(\HferC^1)\cong \Lambda(\overline  {PL^2(S^1)} ) \otimes \Lambda(P^\perp L^2(S^1))$
which is isotonic by definition and fulfills twisted duality, namely by Haag-Araki duality
$\FerC(I')=C(K(I)^\perp)=C(K(I))^\perp = \FerC(I)^\perp$.
In addition, the net $\FerC$ is M\"obius covariant. Indeed, we can take the representation
$\Lambda U(\cdot)$ by promoting the one-particle representation $U$ to the second
quantization operator. It is easy to see that the covariance of this net $\FerC$
follows from the covariance of the net of standard spaces $K$.
The representation $\Lambda U$ has positive energy since so does
the representation $U$, and leaves invariant the vacuum vector $\Omega_0$ of the Fock space.
Summing up, the net $\FerC$ is a fermi net (cf. \cite{Wassermann98}). This net is referred to as
the {\bf free complex fermi net} on $S^1$.
The scalar multiplication by a constant phase $\ee^{-\ima \vartheta}$
in the original structure of the one-particle space
is still a unitary operator in the new structure. Its promotion
by the second quantization $V(\theta)$ implements an action
of $\uone$ on $\FerC$ by inner symmetry.
This will be referred to as the $\uone$-gauge action.

For $r\in \frac12 +\ZZ$ let $\psi_r=a_P(e_{-r-\frac12})$
and $\bar\psi_r=a_P(e_{r-\frac12})^\ast$ where $e_r \in L^2(S^1)$ with $e_r(\theta)=\ee^{\ima \theta r}$. The 
$\psi_r,\bar\psi_r$ are the modes of the free complex fermion, namely
\begin{align*}
    \{\psi_n,\psi_m\} &= \{\bar\psi_m,\bar\psi_n \}=0\\
    \{\bar\psi_n,\psi_m\} &= \delta_{m+n,0}\\ 
    \psi_n^\ast &=\bar\psi_{-n}
\end{align*}
and it holds that $\psi_r \Omega_0= \bar\psi_r\Omega_0=0$ for $r\in \frac12+\NN_0$.
Each of $\psi_r$ or $\bar\psi_r$ has norm $1$ following
from the commutation relation.
We can introduce the usual fields ($f,g\in L^2(S^1)$) and operator valued
distributions in the $z$-picture:
\begin{align*}
    \Psi(f) &= \sum_{r\in \frac12+\ZZ} \hat f_r \Psi_r 
    =\oint_{S^1} f(z) z^{-\frac12} \Psi(z)  \frac{\dd z}{2\pi \ima}\,,
    &\Psi(z) &= \sum_{r\in\frac12+\ZZ} \Psi_r z^{-r-\frac12}\,,
    \\\bar\psi(f)&=\psi(\overline f)^\ast = a_P(e_{-\frac 12} f)^\ast\,,
    & \{\bar \psi(f),\psi(g)\} &= \oint f(z) g(z) \frac{\dd z}{2\pi \ima z}
    \,,
\end{align*}
where $\Psi$ is either $\psi$ or $\bar\psi$. The fields $\psi,\bar\psi$ are covariant, e.g.
$U(g)\psi(f)U(g)^\ast =\psi(f_g)$ with 
$f_g(z) = \frac{1}{|\a-\overline \b z|}f(g^{-1}z)$ for $g=\left(\begin{matrix} \a & \b \\ \overline{\b} & \overline{\a}\end{matrix}\right) \in \mathrm{SU(1,1)}$.

We note that vectors of the form
\[
    \xi=\psi_{-r_1}\cdots\psi_{-r_k}\bar\psi_{-s_1}\cdots \bar\psi_{-s_\ell} \Omega_0 
\]
with $0<r_1<\cdots<r_k$ and $0<s_1<\cdots <s_\ell$
form a basis of $\HferC=\Lambda (\HferC^1)$ and that such a $\xi$ is an 
eigenvector for the rotations, $R(\theta)\xi\equiv \ee^{\ima\theta L_0}\xi=\ee^{\ima N\theta}\xi$
with $N=\sum_{j=1}^k r_j + \sum_{j=1}^\ell s_j$ and of 
the gauge action $V(\theta)\xi \equiv \ee^{\ima \theta Q}\xi = \ee^{\ima (k-\ell)\theta }\xi$.
In each vector of this basis the $r$-th energy level can either be empty, be occupied by $\psi_{-r}$ or $\bar\psi_{-r}$ 
or occupied by both. The contribution of this level to the character  $\tr_{\HferC}(\ee^{-\beta L_0-EQ})$ is then $1$, 
$zt^r$, $z^{-1}t^r$ or $t^{2r}$, respectively, where $t=\ee^{-\beta}$ and $z=\ee^{-E}$. 
By summing over all possibilities one gets that the character of $\FerC$ is given by (cf. \cite{Kac1998,RehrenCQFT}):
\begin{align*}
    \tr_{\HferC}(\ee^{-\beta L_0-EQ})=\tr_{\HferC}(t^{L_0}z^Q) &= \prod_{r\in\NN_0+\frac12} (1+zt^r+z^{-1}t^r + t^{2r})
                    \\&= \prod_{r\in\NN_0+\frac12} (1+zt^r)(1+z^{-1}t^r)
                    \\&= p(t)\sum_{q\in\ZZ}z^q t^{\frac{q^2}2} \,,
\end{align*}
where the last equality follows directly from the Jacobi triple product formula
(see \cite[Theorem 14.6]{Apostol76})
\begin{align*}
    \prod_{r\in\NN}(1+zw^{2r-1})(1+z^{-1}w^{2r-1})(1-w^{2r})=\sum_{q\in\ZZ}z^qw^q
\end{align*}
by setting $2r-1=2n$ and $t=w^2$. In particular, for the local net
$\FerC^{\uone}$ the character is given by
$\tr_{\HferC^{\uone}}(\ee^{-\beta L_0}) = p(t)$, since it is the
fixed point with respect to the $\uone$-gauge action
and the conformal character is the coefficient of $z^0$.

\subsection{\texorpdfstring{$\uone$}{U(1)}-current net as a subnet of \texorpdfstring{$\FerC$}{FerC}}\label{u1-current-subnet}
In this section we use the well-known fact that the Wick product $:\!\bar\psi\psi\!:$ of the complex fermion $\psi$
equals the $\uone$-current and give an analogue of the boson-fermion correspondence (see e.g. \cite[5.2]{Kac1998}) in the operator algebraic setting.
Let us denote by $\D_0$ the subspace of $\Lambda(\H_P^1)$ of vectors with finite energy:
\[
\D_0 := \mathrm{span}\left\{\psi_{-r_1}\cdots \psi_{-r_k}\bar\psi_{-s_1}\cdots \bar\psi_{-s_l}\Omega_0:
 k,l \in \NN_0, r_i,s_j \in \NN+\frac 12\right\}
 \,.
\]
Then we define the unbounded operators on the domain $\D_0$:
\begin{eqnarray*}
    J_n =\sum_{r+s=n} {:\!\bar\psi_r \psi_s\!:}
&=& \sum_{r<0} \bar\psi_{r} \psi_{n-r} - \sum_{r>0} \psi_{n-r}\bar\psi_{r} \\
&=&  \sum_r\left( \bar\psi_{r} \psi_{n-r} - \<\Omega_0,\bar\psi_{r} \psi_{n-r}\Omega_0\>  \right)
\end{eqnarray*}
with $r,s\in \frac 12+\ZZ$. Note that any vector in $\D_0$ is annihilated by $\psi_r$ for
sufficiently large $r$, thus the action of $J_n$ on such a vector can be defined and
remains in $\D_0$. In particular, we have $J_n\Omega_0=0$ for $n\in\NN_0$.
\begin{lemma}\label{lm:commutation-field-current}
On $\D_0$ it holds that
    \begin{enumerate}
        \item 
            $[J_n,\psi_k] =-\psi_{n+k}$ and $[J_n,\bar\psi_k]=\bar\psi_{n+k}$
        \item $[J_m,J_n]= m\delta_{m+n,0}$
    \end{enumerate}
\end{lemma}
\begin{proof}
Using $ [ab,c]=a\{b,c\}-\{a,c\}b$, one obtains
    $[\bar\psi_r\psi_n,\psi_k] = -\delta_{r+k,0}\psi_n$ and $[\psi_n\bar\psi_r,\psi_k] = \delta_{r+k,0}\psi_n$ from
    which directly follows 
$[J_n,\psi_k]= \sum_{r<0} [\bar\psi_{r} \psi_{n-r},\psi_k] - \sum_{r>0} [\psi_{n-m}\bar\psi_{r},\psi_k]
    =-\psi_{n+k}$. Analogously one shows $[J_n,\bar\psi_k]=\bar\psi_{n+k}$.

From the Jacobi identity,
it follows immediately that $[J_n ,J_m]$ commutes with all $\psi_k$ and $\bar\psi_k$
and hence $[J_n,J_m]$ is a multiple of the identity,
therefore $[J_n,J_m]= \<\Omega_0,[J_n,J_m]\Omega_0\>\1$. It is
\begin{align*}
    [J_n,J_p] 
    &= \sum_{r<0} [J_n,\bar\psi_{r} \psi_{p-r}] - \sum_{r>0}[J_n ,\psi_{p-r}\bar\psi_{r}]
    \\ &=- \sum_{r<0} \left(\bar\psi_{r} \psi_{p-r+n}- \bar\psi_{r+n} \psi_{p-r}    \right)
    -\sum_{r>0}\left(\psi_{p-r}\bar\psi_{r+n}-\psi_{p-r+n}\bar\psi_{r}  \right)
\end{align*}
and in the case $p\neq -n$ we get $\<\Omega_0,[J_n,J_{p}]\Omega_0\>=0$, and otherwise
\begin{align*}
    \<\Omega_0,[J_n,J_{-n}]\Omega_0\> &=
        \begin{cases}
            \sum_{r<0}\<\Omega_0,\bar\psi_{r+n} \psi_{-r-n}\Omega_0\>= 
                \sum_{r=\frac12}^{n-\frac12}\<\Omega_0,\{\bar\psi_{r},\psi_{-r}\}\Omega_0\>  &n>0\\
            -\sum_{r>0}\<\Omega_0,\psi_{-r-n}\bar\psi_{r+n} \Omega_0\>= 
                -\sum_{r=\frac12}^{-n-\frac12}\<\Omega_0,\{\psi_{r},\bar\psi_{-r}\}\Omega_0\>  &n<0
        \end{cases}
        \\ &=n \,,
\end{align*}
which completes the proof.
\end{proof}

Let $L_0$ be the generator of the rotation: $R(\theta) = \ee^{\ima\theta L_0}$.
From its action (see the end of Section \ref{free-fermion}), one verifies that
$\D_0$ is a core for $L_0$.
\begin{lemma}[Linear energy bounds]\label{lm:energy-bound} It holds that $[L_0,J_n] = -n J_n$ on $\D_0$.
For a trigonometric polynomial $f = \sum_n \hat{f}_ne_n$ where the sum is finite and 
$\xi\in \D_0$, we have
\begin{align*}
    \|J(f)\xi\| &\leq c_f\|(L_0+1)\xi\|
    \\ \|[L_0,J(f)]\xi\| &\leq c_{\partial_\theta f}\|(L_0+1)\xi\|,
\end{align*}
where $c_f$ depends only on $f$.
\end{lemma}
\begin{proof}
For the commutation relation,
it is enough to choose an energy eigenvector $\xi\in\D_0$, i.e. $L_0\xi =N\xi$.
It is
$J_n L_0\xi = N J_n\xi$ and
 \[
L_0 J_n \xi = L_0\left(\sum_{r< 0} \bar\psi_r\psi_{n-r}\xi - \sum_{r>0} \psi_{n-r}\bar\psi_{r}\xi\right)
= (N-n)J_n \xi,
\]
and the first statement follows.

We have seen that $\psi_r$ and $\bar\psi_r$ have norm $1$ in Section \ref{free-fermion}.
First we claim that $\|J_n\xi\| \le \|(2(L_0+1)+|n|)\xi\|$.
Let $\xi$ be again an eigenvector of $L_0$, i.e. $L_0\xi = N\xi$.
From the defining sum of $J_n$, one sees that
only $2N+|n|+2$ terms contribute to $J_n\xi$.
Hence we have $\|J_n\xi\|\leq (2N+|n|+2)\|\xi \|  = \| 2(L_0+1)+|n|\xi\|$.
If the inequality holds for eigenvectors, then
for $\{\xi_r\}$ with different eigenvalues, we have
$\xi_r\perp\xi_s$ and $J_n\xi_r\perp J_n\xi_s$, and hence
\begin{eqnarray*}
\left\|J_n\sum_r \xi_r\right\|^2 &=& \sum_r\|J_n\xi_r\|^2 \\
&\le & \sum_r\|(2(L_0+1)+|n|)\xi_r\|^2 \\
&=& \left\|(2(L_0+1)+|n|)\sum_r\xi_r\right\|^2
\end{eqnarray*}
and the general case follows.

For a smeared field, we have
\[
\|J(f)\xi\| = \left\|\sum_n \hat{f}_n J_n\xi\right\| \le
2\tilde c_f\|(L_0+1)\xi\|+\tilde c_{\partial_\theta f}\|\xi\| \le
(2\tilde c_f+\tilde c_{\partial_\theta f})\|(L_0+1)\xi\|,
\]
where
$\tilde c_f = \sum_n |\hat{f}_n|$. By defining
$c_f = 2\tilde c_f + \tilde c_{\partial_\theta f}$, we obtain
the first inequality of the statement.
The rest follows by noting that $[L_0,J(f)] = J(\ima \partial_\theta f)$.
\end{proof}

For a smooth function $f = \sum_{n\in\ZZ} \hat f_n e_n \in C^\infty(S^1)$,
its Fourier coefficients $\hat{f}_n$ are strongly decreasing and, in particular,
it is summable: $\sum_n |\hat{f}_n| = \tilde c_f < \infty$.
Hence we can naturally extend the definition of the smeared current to smooth functions
using the above estimate by
\[
J(f)  = \sum_{n\in\ZZ} f_{n} J_n = \sum_{r,s\in\frac12+\ZZ} f_{r+s} :\!\psi_r\bar \psi_s\!:,
\]
and the same inequality in Lemma \ref{lm:energy-bound} holds.
The operator is closable since we have $J(f) \subset J(\overline{f})^*$
and we still denote the closure by $J(f)$.
We note that from the above definition it follows that $J(f)$ is obtained by a limit $\sum_n:\!\psi(h_n)\bar\psi(k_n)\!:$ 
with suitable functions such that $\sum_n h_n(\theta) k_n(\vartheta) \to 2\pi f(\theta) \delta(\theta-\vartheta)$.
This implies covariance of the ``field'', i.e. $U(g)J(f)U(g)^\ast=J(f\circ g^{-1})$.

Recall that $\|\psi_r\| = 1$, hence the smeared field is still bounded:
$\|\psi(g)\| \le \tilde c_g$.
We claim that, for $f,g\in C^\infty(S^1)$ and $\xi\in\D_0$,
$\psi(g)\xi$ is in the domain of $J(f)$. Indeed, for a trigonometric polynomial $g$,
we have the estimate
\begin{eqnarray*}
\|J(f)\psi(g)\xi\| &\le& c_f\|(L_0+1)\psi(g)\xi\| \\
&\le& c_f(\tilde c_g \|\xi\| + \|[L_0,\psi(g)]\xi+\psi(g)L_0\xi\|) \\
&\le& c_f(\tilde c_g(\|\xi\|+\|L_0\xi\|)+\tilde c_{\partial_\theta g}\|\xi\|)\,.
\end{eqnarray*}
Then if we have a sequence of trigonometric polynomial $g_n$ converging to
a smooth function $g \in C^\infty(S^1)$, the sequence $\{J(f)\psi(g_n)\xi\}$ is
also converging.

\begin{lemma}\label{lm:commutation-smeared}
For $\xi,\eta \in \D_0$, it holds that
\begin{align*}
[J(f),\psi(g)] \xi &= - \psi(f\cdot g)\xi \\
[J(f),\bar\psi(g)]\xi &= \bar\psi(f\cdot g)\xi \\
\<J(\bar f)\xi,J(g)\eta\> &= \<J(\bar g)\xi,J(f)\eta\> + 2\ima\omega(f,g)\<\xi,\eta\>.
\end{align*}
\end{lemma}
\begin{proof}
For trigonometric polynomials $f,g$,
the statements can be proved easily from
Lemma \ref{lm:commutation-field-current}. The general case is shown
by approximating first $f$ by polynomials, then $g$, according to the
convergence considered above (as for the third statement,
obviously the order of limits does not matter).
\end{proof}

We need the following well-known result \cite[Theorem 3.1]{DF77}:
\begin{theorem}[The commutator theorem]
Let $H$ be a positive self-adjoint operator and $A,B$ symmetric operators
defined on a core $\D_0$ for $(H+\1)^2$.
Assume that there is a constant $C$ such that
\begin{gather*}
\|A\xi\| \le C\|(H+\1)\xi\|, \,\,\, \|B\xi\| \le C\|(H+\1)\xi\|, \\
\|[H,A]\xi\| \le C\|(H+\1)\xi\|, \,\,\, \|[H,B]\xi\| \le C\|(H+\1)\xi\|, \\
\<A\xi,B\eta\> = \<B\xi,A\eta\> \mbox{ for any } \xi,\eta \in \D_0.
\end{gather*}
Then $A$ and $B$ are essentially self-adjoint on any core of $H$ and
any bounded functional calculus of $A$ and $B$ commute.
\end{theorem}
\begin{remark}
In the original literature \cite{DF77}, this Theorem is proved under the
assumption of certain operator inequalities. In fact, what is really
used in the proof of commutativity of bounded functions is the norm estimates
$\|A(H+\1)^{-1}\| < C,\|[H,A](H+\1)^{-1}\|<C$ etc.\! and
they follow from the assumptions here.
The essential self-adjointness of $A$ and $B$ can be proved
by \cite[Theorem X.37]{RSII}.
An analogous application of this theorem with norm estimates
can be found in \cite{BS90}.
\end{remark}
By the commutator theorem, we get 
that $J(f)$ is self-adjoint for $f\in C^\infty(S^1,\RR)$  and that all bounded functions of $J(f)$ 
commute with all bounded functions of $J(g)$
for $f,g\in C^\infty(S^1,\RR)$ with disjoint support. 

Let $I$ be a proper interval and let us define the von Neumann algebra
\[
\B(I)  = \{ \ee^{\ima J(f)} : \supp f \subset I\}''.
\]
The local net $\B(I)$ restricted to
 $\overline{\B(I)\Omega_0}$ can be identified with the $\uone$-current net $\u1net$ on 
 $\Hbose$, in particular we can identify  $\overline{\B(I)\Omega_0}\cong \Hbose$.
\begin{proposition} Let $I$ be a proper interval, then 
    $\B(I) \subset \FerC^{\uone}(I)$. 
\end{proposition}
\begin{proof}
We see that $\B(I)$ commutes with $\FerC(I') = \{c(g): g\in L^2(I') \}''$
because, for $f,g$ with disjoint supports,
$c(g)$ commutes with $J(f)$ on a core
by Lemma \ref{lm:commutation-smeared}
and therefore any spectral projection of $c(g)$
commutes with $J(f)$, and hence with any bounded functions of $J(f)$.

Further because $J(f)$ commutes by construction with the gauge action $V(t)$ and is in particular 
even because $V(\pi)=\Gamma$, it follows that $\B(I)$ lies in the twisted commutant
$\FerC(I')^\perp$. By twisted Haag duality it is $\B(I)\subset \FerC(I')^\perp = \FerC(I)$ and therefore
$\B(I) = \B(I)^{\uone}  \subset \FerC^{\uone}(I)$.
\end{proof}

Since the covariance has been seen, we have the following.
\begin{corollary} $\B$ is a subnet of $\FerC^{\uone}$. 
\end{corollary}

Now the following is straightforward.
\begin{proposition}
The $\uone$-fixed point subnet of the complex free fermion net $\FerC$ is the $\uone$-current net,
i.e. $\FerC^{\uone} = \B \cong \u1net$. 
\end{proposition}
\begin{proof}
Let us see $\B$ as a subnet of the fermi net $\FerC^{\uone}$ 
on $\HferC^{\uone} \equiv \H_{\,\cdot\,,0}$.
Further $\overline{\B(I)\Omega}$ does not depend on $I$
by the same proof of the Reeh-Schlieder property
and is clearly a subspace of $\HferC^{\uone} \equiv \H_{\,\cdot\,,0}$.

In fact they coincide, since we have confirmed that
$\tr_{\Hbose}(\ee^{-\beta L_0})=\tr_{\H_{\,\cdot\,,0}}(\ee^{-\beta L_0}) = p(t)$,
where $\ee^{-\beta} = t$, namely, their conformal characters
coincide (see also Section \ref{subnets}).
\end{proof}

We finish this section by giving the parametrization in $x$-picture, where
 the action of the translation is more natural. With 
\[
    f(x) 
    = \frac{1}{\sqrt{2\pi}} \sqrt{\left| \frac{\partial \theta(x)}{\partial x}\right| } \ee^{\ima \theta(x)/2} f_0(\theta(x))
\]
we identify $L^2(\RR)=L^2(\RR,\dd x)$ with $L^2(S^1)=L^2([0,2\pi],\dd\theta /(2\pi))$ and therefore 
the space $\HferC^1$ is given by $PL^2(\RR) \oplus \overline {P^\perp L^2(\RR)}$ with 
$P:\hat f(p) \mapsto \Theta(p)\hat f(p)$ and it can be identified in ``momentum space'' with 
$L^2(\RR_+,2\pi\dd p)\oplus L^2(\RR_+,2\pi \dd q)$  by
\begin{align*}
    f(x) &\longmapsto \widehat {Pf}(p) \oplus  \overline{\widehat{P^\perp f}(-q)} & p,q>0
    \,.
\end{align*}

The field operators are defined for $f\in L^2(\RR)$ by 
$\psi(f)=a_P(f)$ and $\bar\psi(f)=a_P(\overline f)^\ast$. For $\Psi \in \HferC$ 
we write its components 
\begin{equation}
    \label{eq:decomposition}
    \Psi_{m,n} \in \H_{m,n}:=L^2(\RR_+^{m+n},(2\pi)^{m+n}\dd p_1\cdots \dd p_m\dd q_1 \cdots q_n)_-,
\end{equation}
where $-$ means the antisymmetrization within $p_1,\ldots, p_m$ and $q_1,\ldots, q_n$.
By this notation $(\psi(f)\Omega_0)_{1,0}(p)= \hat f(p)$ and $(\bar\psi(f)\Omega_0)_{0,1}(q)=\hat f(q)$.
Further the bi-field
${:\!\bar\psi(f)\psi(g)\!:}=\bar\psi(f)\psi(g)-\<\Omega_0,\bar\psi(f)\psi(g)\Omega_0\> \1 $ creates from the vacuum $\Omega_0$
a fermionic 1+1 particle state $\Psi_{f,g}:={:\!\bar\psi(f)\psi(g)\!:}\Omega_0$ 
with $(\Psi_{f,g})_{1,1}(p,q) =-\hat f(q)\hat g(p)$
and it follows for $h\in C^\infty(\RR,\RR)$ that for the $\uone$-current $J$, it holds
$(J(h)\Omega_0)_{1,1}(p,q) =-\frac1{2\pi} \hat h(p+q)$ which is obtained by taking 
a limit $\sum_n\Psi_{f_n,g_n}$ with test functions $\sum_n f_n(x)g_n(y) \to h(x)\delta(x-y)$. 
We make the important observation that the $J(f)\Omega_0$ generate the one-particle 
space which we can identify with $\Hbose^1$ and this is obviously a proper subspace of the 
fermionic 1+1-particle space $\HferC^{1,1}$. 

\section{A new family of Longo-Witten endomorphisms on
\texorpdfstring{$\uone$}{U\^{}1}-current net}\label{endomorphisms}

We use the description of $\HferC^1=\overline{PL^2(S^1)} + P^\perp L^2(S^1)$ which equals 
$L^2(S^1)$ as a real Hilbert space and is described in the beginning 
of Section \ref{free-fermion}.
First we decompose $\HferC^1$ into irreducible representations of $\mathrm{SU}(1,1)$ in a compatible
way with $K(I)$. Let us define
\begin{eqnarray*}
\H_\Re &:=& \{ f \in \H^1_{\FerC}: z^{\frac{1}{2}}f(z) \mbox{ is real}\}, \\
\H_\Im &:=& \{ f \in \H^1_{\FerC}: z^{\frac{1}{2}}f(z) \mbox{ is pure imaginary}\}.
\end{eqnarray*}
By their definition, it is clear that $\H_\Re$ and $\H_\Im$ are real Hilbert subspaces of $L^2(S^1)$.
In fact, they are complex subspaces with respect to the new complex structure.
To see this, we take another description of $\H_\Re$:
in terms of Fourier components, it holds that $f \in \H_\Re$ if and only if
$f_n = \overline{f_{-n-1}}$. Recall that, on $L^2(S^1)$, the new scalar multiplication by $\ima$ is given by
$\ima\cdot f_n = -\ima f_n$, $\ima\cdot f_{-n-1} = \ima f_{-n-1}$ for $n \ge 0$. Hence this condition
is preserved under the multiplication by $\ima$ and $\H_\Re$ is a complex subspace. An analogous
argument holds for $\H_\Im$. Next we see that $\H_\Re$ and $\H_\Im$ are orthogonal.
Note that because of the change of the complex structure, for $f(z) = \sum_n f_n z^n$ and
$h(z) = \sum_n h_n z^n$ the inner product is written as follows:
\[
\<f,h\> = \sum_{n \ge 0} f_n \overline{h_n} + \sum_{n < 0} \overline{f_n} h_n.
\]
Now $f \in \H_\Re$ implies $f_n = \overline{f_{-n-1}}$ and $h \in \H_\Im$ implies
$h_n = -\overline{h_{-n-1}}$ for non-negative $n$, hence it is easy to see that
\[
\<f,h\> = \sum_{n \ge 0} f_n\overline{h_n} + \sum_{n < 0} \overline {f_n} h_n
= -\sum_{n \ge 0} \overline{f_{-n-1}}{h_{-n-1}} + \sum_{n < 0} \overline{f_n}{h_n} = 0.
\]
In other words, these two complex subspaces are mutually orthogonal.

Furthermore, $\H_\Re$ and $\H_\Im$ are invariant under the action of $\mathrm{SU}(1,1)$.
We recall that the action is given by
$(V_g f)(z) = \frac{1}{-\overline{\b}z+\a}f\left(\frac{\overline{\a}z-\b}{-\overline{\b}z+\a}\right)$.
Then if $z^\frac{1}{2}f(z)$ is real then it holds that
\begin{eqnarray*}
z^{\frac 1 2}(V_g f)(z)
&=& \frac{1}{{\bar z}^{\frac 1 2}(-\overline{\b}z+\a)}\left(\frac{\overline{\a}z-\b}{-\overline{\b}z+\a}\right)^{-\frac{1}{2}}
\cdot\left(\frac{\overline{\a}z-\b}{-\overline{\b}z+\a}\right)^\frac{1}{2}
f\left(\frac{\overline{\a}z-\b}{-\overline{\b}z+\a}\right) \\
&=& \frac{1}{(-\overline{\b}+\a \bar z)^\frac{1}{2}(\overline{\a}z-\b)^\frac{1}{2}}
\cdot\left(\frac{\overline{\a}z-\b}{-\overline{\b}z+\a}\right)^\frac{1}{2}
f\left(\frac{\overline{\a}z-\b}{-\overline{\b}z+\a}\right)
\end{eqnarray*}
and both factors are real. Similarly one shows that $\H_\Im$ is preserved under $V_g$.
It is obvious that these two representations are intertwined by the multiplication by $\ima$
in the old complex structure. This is still a unitary map, thus they are unitarily equivalent.
One can see that each representation is indeed irreducible, and when restricted to
$\mathrm{PSU}(1,1) = \psl2r$, it is the projective positive energy representation with lowest weight $\frac{1}{2}$.

It is easy to see that $e^\Re_n := \{e_n + e_{-n+1}, n \ge 0\}$ and $e^\Im_n := \{\ima (e_n + e_{-n+1}), n \ge 0\}$ form
bases of $\H_\Re$ and $\H_\Im$, respectively, where $e_n(z) = z^n$ and the multiplication by $\ima$ is
given in the old structure.
Now we describe the gauge action in terms of this basis. By the definition, for a given complex number $\a$ with modulus $1$,
the action is given by the multiplication in the old structure. Hence if $\a = \cos \theta +\ima \sin \theta$,
we have $U_\a e^\Re_n = \cos \theta e^\Re_n + \sin \theta e^\Im_n$ and
$U_\a e^\Im_n = -\sin \theta e^\Re_n + \cos \theta e^\Im_n$.
This means that $U_\a$ acts as the real rotation by $\theta$ in this basis.

\subsubsection*{Construction of endomorphisms}
We construct Longo-Witten endomorphisms on the free fermion net $\FerC$ commuting with the gauge action. 
The key is the following theorem. We remind that
a \textbf{standard pair} $(\tilde H,\tilde T)$ is a standard subspace $\tilde H\subset \tilde \H$ of a Hilbert 
space $\tilde \H$ and a positive energy representation $\tilde T$ of $\RR$ on $\tilde \H$, 
such that $\tilde{T}(a)\tilde{H} \subset \tilde{H}$ for $a \ge 0$.
If $\tilde T$ is maximally abelian, the standard pair is said to be irreducible and there is a (up to unitary equivalence) 
unique irreducible standard pair.

\begin{theorem}[{\cite[Theorem 2.6]{LW11}}]
Let $(\tilde{H},\tilde{T})$ be a standard pair with
multiplicity $n$, i.e. it decomposes into $n$-fold direct sum 
of irreducible standard pairs, each unitarily equivalent to the unique standard pair $(H,T)$ and $T(t)=\ee^{\ima t P}$.
Then a unitary $\tilde{V}$ commuting with the translation $\tilde{T}$
preserves $\tilde{H}$ if and only if $\tilde{V}$
is a $n\times n$ matrix $(V_{hk})$ (with respect to the decomposition
of $\tilde\H$ into $n$ direct sum as above) such that $V_{hk} = \f_{hk}(P)$,
where $\f_{hk}: \RR \to\CC$ are complex Borel functions such that
$(\f_{hk})$ is a unitary matrix for almost every $p > 0$, each $\f_{hk}$ is the
boundary value of a function in $\HH(\SS_\infty)$ and is symmetric, i.e. $\f_{hk}(-p)=\overline{\f_{hk}(p)}$.
\end{theorem}

Consider the one-particle space $\HferC^1$ for $\FerC$.
The pair of the standard space $K(\RR_+) = L^2(\RR_+)$ defined in Section \ref{free-fermion}
(under the identification of $S^1$ and $\RR \cup \{\infty\})$
and the natural translation has multiplicity 2.
If we take a matrix-valued function $(\f_{hk})$ as above and take the second quantization operator
$\Lambda(V)$ of the (matrix-valued) operator $(V_{kh}) = (\f_{hk}(P))$, then it implements a
Longo-Witten endomorphism of $\FerC$ (see \cite{LW11}).

As the gauge group acts by real rotation
$\left(\begin{matrix}\cos\theta & -\sin\theta \\ \sin\theta & \cos\theta\end{matrix}\right)$,
any matrix-valued function of $p$ which commute with them must have the form
$\left(\begin{matrix}a(p) & \ima b(p) \\ -\ima b(p) & a(p)\end{matrix}\right)$.
If each component is symmetric, then $a$ is symmetric and $b$ is antisymmetric.
Such a matrix-valued function can be diagonalized by the matrix
$\left(\begin{matrix} 1 & \ima \\ \ima & 1\end{matrix}\right)$ and becomes
$\left(\begin{matrix} a(p)+b(p) & 0 \\ 0 & a(p)-b(p) \end{matrix}\right)$.
We claim that such $a$ and $b$ exist. Indeed, let $\f$ be a inner function
(not necessarily symmetric), namely the boundary value with modulus $1$ of a
bounded analytic function on the upper half-plane $\HH$, and define
$a(p) = \frac{1}{2}(\f(p) + \overline{\f(-p)})$, $b(p) = \frac{1}{2}(\f(p)-\overline{\f(-p)})$.
Then it is obvious that $a$ is symmetric and $b$ is antisymmetric.
In addition, $a(p) + b(p) = \f(p)$ and $a(p) - b(p) = \overline{\f(-p)}$, hence
the diagonalized matrix is unitary for almost every $p$.
By the theorem of Longo-Witten, the operator
$\f(P_1) := \left(\begin{matrix}a(P) & \ima b(P) \\ -\ima b(P) & a(P)\end{matrix}\right)$
preserves the real Hilbert space $\widetilde{H}:=K(\RR_+)$, where
$P_1$ is the generator of the translation in $\H_{\FerC}^1$ which has
multiplicity $2$ and
$P$ is the generator of $T$ of the irreducible standard pair $(H,T)$.

It is easy to see that the above diagonalization is given exactly by the decomposition
$\HferC^1 = \overline{PL^2(S^1)} \oplus (\1-P)L^2(S^1)$.

We remind that $\HferC^1$ can be identified with $L^2(\RR)$ as a real space. 
In $L^2(\RR)$ the function $\f(P_1)f$ is the function with Fourier
transform $\f(p)\hat f(p)$ and we remark that it also follows directly 
from the Paley-Wiener theorem that $\f(P_1)$ leaves $L^2(\RR_+)\subset L^2(\RR)$ invariant for 
$\f$ inner.
Further using that the space $\HferC$ decomposes in 
$\HferC = \bigoplus_{m,n\in \NN_0} \H_{m,n}$ like in (\ref{eq:decomposition})
with the gauge action given by $V(\theta)\Psi_{m,n} = \ee^{\ima (m-n)\theta}\Psi_{m,n}$,
the action of the Longo-Witten unitary $V_\f=\Lambda(\f(P_1))$ is given by 
\begin{align*}
    &(V_\f\Psi)_{m,n}(p_{1},\cdots,p_{m},q_{1},\cdots,q_{n}) \\
    &\quad = 
    \f(p_1)\cdots\f(p_m)\overline{\f(-q_1)}
    \cdots \overline{\f(-q_n)}
    \Psi_{m,n}(p_{1},\cdots,p_{m},q_{1},\cdots,q_{n}) \,.
\end{align*}

\begin{lemma}\label{lem:LW1p} 
    Let $\iota:\Psi\in L^2(\RR_+, p\dd p)\equiv\Hbose^1 \hookrightarrow \H_{1,1} \subset \HferC$ be the embedding 
    given by $\iota(\Psi)_{1,1}(p,q) = -\frac1{2\pi}\Psi(p+q)$. A second quantization Longo-Witten unitary $V_\f$
    commuting with the gauge action $V(\slot)$ satisfies  $V_\f\iota\Hbose^1\subset \iota\Hbose^1$ if and only 
    if $V_\f = V(\theta)T(t)$ with $t\geq 0$.
\end{lemma}
\begin{proof}
    The translations commute with the gauge action and it follows immediately that they leave $\Hbose^1$ invariant.
    We note that $\f(p)\overline{\f(-q)}\Psi(p+q)$ belongs to $\Hbose^1$ only if it can be written
    as a function of $g(p+q)$. This means that $\f(p)\overline{\f(-q)}=\widetilde\f(p+q)$ for $p,q \ge 0$, where $\widetilde\f$ is another
    function.
Then, putting $q = 0$ and $p = 0$ respectively,
we see that $\f(p)\overline{\f(0)} = \widetilde\f(p)$ for $p \ge 0$ and 
$\f(0)\overline{\f(-q)} = \widetilde\f(q)$ for $q \ge 0$, in particular
$\widetilde\f(0) = 1$. Multiplying the each side of these equations, one sees that
$\widetilde\f(p+q) = \widetilde\f(p)\widetilde\f(q)$ because $|\f(0)| = 1$. Then it follows that
$\widetilde\f(p) = \ee^{\ima \k p}$ for some $\k \ge 0$,
and $\f(p) = \ee^{\ima (\k p+\theta)}$
for some $\theta \in \RR$ (in fact, the arguments here should be
treated with care because the relation is given only almost everywhere,
but both $\f$ and $\check \f$ analytically continue and
in the domain of analyticity it holds everywhere).

Such a $\f$ is a Longo-Witten unitary only for $\k\geq 0$.
The constant factor $\ee^{\ima \theta}$ corresponds to the factor
$V(\theta)$.
\end{proof}

\begin{theorem}\label{th:longo-witten-unitaries} Let $\f$ be an inner function as above.
The endomorphism implemented by the second quantization $V_\f$ of the operator constructed above
restricts to the $\uone$-current subnet.
The restriction cannot be implemented by any second quantization operator
if $\f(p) \neq \ee^{\ima (\k p+\theta)}$.
\end{theorem}
\begin{proof}
The operator $V_\f$ restricts to the subnet $\u1net$ by the general
argument in Proposition \ref{pr:restricting-endomorphisms}.
It cannot be implemented by a second quantization operator,
since any second quantization operator preserves the particle number,
while $V_\f$ does not for non-exponential $\f$ as we saw above,
and a Longo-Witten endomorphism is uniquely
implemented up to scalar (see Section \ref{fermi-nets}).
\end{proof}

\begin{remark} By the construction in \cite{LW11}, each unitary $V=V_\f|_{\Hbose}$ related to an
    inner function $\f$ from above gives rise to a local, 
    time-translation covariant 
    net of von Neumann algebras on the Minkowski half-space $M_+=\{(t,x)\in \RR^2:x>0\}$. This
    net is associated with the $\uone$-current net 
    $\u1net$ and defined by 
    $\u1net_V(O) = \u1net(I_1) \vee V\u1net(I_2)V^\ast$, where $O=I_1\times I_2=\{(t,x)\in\RR^2:t-x\in I_1, t+x\in I_2\}$ is a double cone with  
    $\overline O \subset M_+$ corresponding uniquely to the two intervals $I_1$ and $I_2$ with disjoint  closures.
    In the case where $\f$ is not exponential $V_\f$ does not come from second quantization---in contrast to the unitaries
    constructed by Longo and Witten in \cite{LW11}---and therefore gives new examples.
\end{remark}

\section{Interacting wedge-local net with particle production}\label{wedge-local-nets}

\subsection{Construction of scattering operators}\label{scattering-operator}
In the previous section we saw that, in the basis
$\{e_n+e_{-n}, e_n-e_{-n}\}$ the matrix operator
$\left(\begin{matrix}a(P) & \ima b(P) \\ -\ima b(P) & a(P)\end{matrix}\right)$
implements a Longo-Witten endomorphism if $a$ is symmetric and $b$ is antisymmetric,
and after the simultaneous diagonalization it becomes
$\left(\begin{matrix}\f(P) & 0 \\ 0 & \check{\f}(P)\end{matrix}\right)$ where $\f$ is an
inner function and $\check{\f}(p) = \overline{\f(-p)}$
(note that if $\f$ extends to an analytic function $\f(z)$ on $\HH$,
then $\check{\f}(z) = \overline{\f(-\overline{z})}$ also extends to $\HH$, hence
$\check{\f}$ is again an inner function). By the same argument one sees that
$\left(\begin{matrix}\check{\f}(P) & 0 \\ 0 & \f(P)\end{matrix}\right)$ implements an endomorphism
since $\check{\check{\f}} = \f$.

With respect to the basis after diagonalization, we split the Hilbert space
$\H_{\FerC}^1 =: \H_+ \oplus \H_-$ and the generator
of translation $P_1 =: P_+\oplus P_-$.
Then the tensor product space can be written as follows:
\[
\H_{\FerC}^1\otimes \H_{\FerC}^1 = \left(\H_+\otimes\H_+\right) \oplus \left(\H_+\otimes\H_-\right) \oplus
\left(\H_-\otimes\H_+\right) \oplus \left(\H_-\otimes\H_-\right) \,.
\]
According to this decomposition into a direct sum of four subspaces,
we define an operator
\[
M_\f := \f(P_+\otimes P_+) \oplus \check{\f}(P_+\otimes P_-) \oplus
\check{\f}(P_-\otimes P_+) \oplus \f(P_-\otimes P_-)\,.
\]
Then this restricts to the subspace $\H_{\FerC}^1\otimes \H_+ = \left(\H_+\oplus\H_-\right)\otimes \H_+$
and it is $\f(P_+\otimes P_+)\oplus\check{\f}(P_-\otimes P_+)$, or we can decompose it
with respect to the spectral measure of $P_+$:
\[
\int_{\RR_+} \left(\begin{matrix}\f(pP_+) & 0 \\ 0 & \check{\f}(pP_-)  \end{matrix}\right)\otimes \dE_+(p).
\]
Similarly, the restriction to $\H_{\FerC}^1\otimes \H_-$ is written as
\[
\int_{\RR_+} \left(\begin{matrix}\check{\f}(pP_+) & 0 \\ 0 & \f(pP_-)  \end{matrix}\right)\otimes \dE_-(p).
\]
Using the two-point set $\ZZ_2 = \{+,-\}$ we define
\[
\f_+(p,+) := \f(p),\,\,\,\, \f_+(p,-) = \check{\f}(p),\,\,\,\,
\f_-(p,+) = \check{\f}(p), \,\,\,\, \f_-(p,-) = \f(p).
\]
By defining the spectral measure $E_1 = E_+ \oplus E_-$ on $\H^1$, $M_\f$ can be simply written as
\[
M_\f = \int_{\RR_+ \times \ZZ_2} \left(\begin{matrix}\f_+(pP_+,\i) & 0 \\ 0 & \f_-(pP_-,\i)  \end{matrix}\right)\otimes \dE_1(p,\i),
\]
where $\i = \pm$.

As in \cite{Tanimoto11-3}, we construct the scattering matrix first on the unsymmetrized
Fock space, then restrict it to the antisymmetric space.
For an operator $A$ on $\H_{\FerC}^1\otimes\H_{\FerC}^1$,
we denote by $A^{m,n}_{i,j}$ on
$(\H_{\FerC}^1)^{\otimes m}\otimes(\H_{\FerC}^1)^{\otimes n}$
the operator which acts only on the $i$-th factor
in $(\H_{\FerC}^1)^{\otimes m}$ and $j$-th factor in $(\H_{\FerC}^1)^{\otimes n}$
as $A$. As a convention, $A^{m,n}_{i,j}$ equals to the identity
operator if $m$ or $n$ is $0$.
 Let us denote simply
$\widetilde{\f}(p,\i) := \left(\begin{matrix}\f_+(p,\i) & 0 \\ 0 & \f_-(p,\i)  \end{matrix}\right)$ and
$\widetilde{\f}(P_1,\i) := \left(\begin{matrix}\f_+(P_+,\i) & 0 \\ 0 & \f_-(P_-,\i)  \end{matrix}\right)$.
From the observation above, it is straightforward to see that
\[
(M_\f)^{m,n}_{i,j} = \int
\left(\1\otimes\cdots\otimes \underset{i\mbox{-th}}{\widetilde{\f}(p_jP_1,\i_j)} \otimes\cdots \otimes\1\right)
\otimes \dE_1(p_1,\i_1)\otimes\cdots\otimes \dE_1(p_n,\i_n)
\]
(the case where $m$ or $n$ is $0$ is treated separately).
Then we define, as in \cite{Tanimoto11-3},
\begin{eqnarray*}
S_\f^{m,n} &:=& \prod_{i,j} (M_\f)^{m,n}_{i,j} \\
S_\f &:=& \bigoplus_{m,n} S_\f^{m,n}\,.
\end{eqnarray*}
Let $\H_{\FerC}^\sc$ be the unsymmetrized
Fock space based on $\H_{\FerC}^1$.
Note that $S_\f$ is defined on $\H_{\FerC}^\sc\otimes\H_{\FerC}^\sc$,
and it naturally restricts to $\H_{\FerC}\otimes\H_{\FerC}^\sc$,
$\H_{\FerC}^\sc\otimes\H_{\FerC}$ and $\H_{\FerC}\otimes\H_{\FerC}$.
This $S_\f$ will be interpreted as the scattering matrix. In order to confirm this,
we have to take the spectral decomposition of $S_\f$ only with respect to the
right or left component. Namely,
\begin{eqnarray*}
S_\f &:=& \bigoplus_{m,n} \prod_{i,j} (M_\f)^{m,n}_{i,j} \\
&=& \bigoplus_{m,n} \prod_{i,j} \int
\left(\1\otimes\cdots\otimes \underset{i\mbox{-th}}{\widetilde{\f}(p_jP_1,\i_j)} \otimes\cdots \otimes\1\right)
\otimes \dE_1(p_1,\i_1)\otimes\cdots\otimes \dE_1(p_n,\i_n) \\
&=& \bigoplus_{m,n} \int \prod_{i,j}
\left(\1\otimes\cdots\otimes \underset{i\mbox{-th}}{\widetilde{\f}(p_jP_1,\i_j)} \otimes\cdots \otimes\1\right)
\otimes \dE_1(p_1,\i_1)\otimes\cdots\otimes \dE_1(p_n,\i_n) \\
&=& \bigoplus_n \int \bigoplus_m\prod_j
\left(\widetilde{\f}(p_jP_1,\i_j)\right)^{\otimes m}
\otimes \dE_1(p_1,\i_1)\otimes\cdots\otimes \dE_1(p_n,\i_n) \\
&=& \bigoplus_n \int \prod_j \bigoplus_m
\left(\widetilde{\f}(p_jP_1,\i_j)\right)^{\otimes m}
\otimes \dE_1(p_1,\i_1)\otimes\cdots\otimes \dE_1(p_n,\i_n) \\
&=& \bigoplus_n \int \prod_j \Lambda(\widetilde{\f}(p_jP_1,\i_j))
\otimes \dE_1(p_1,\i_1)\otimes\cdots\otimes \dE_1(p_n,\i_n)\,,
\end{eqnarray*}
where the integral and the product commute in the third equality since the spectral measure
is disjoint for different values of $p$'s and $\i$'s,
and the sum and the product commute in the fifth equality
since the operators in the integrand act on mutually disjoint spaces, namely on $(\H_{\FerC}^1)^{\otimes m}\otimes\H_{\FerC}^\sc$
for different $m$. In the final expression,
all operators appearing in the integrand are the second quantization operators,
thus this formula naturally restricts to the partially antisymmetrized space $\H_{\FerC}\otimes\H_{\FerC}^\sc$.

Now we define
\begin{itemize}
\item $\M_\f := \{x\otimes \1, \Ad S_\f(\1\otimes y): x\in\FerC(\RR_-), y\in\FerC(\RR_+)\}''$,
\item $T(t,x) := T_0(\frac{t-x}{\sqrt 2})\otimes T_0(\frac{t+x}{\sqrt 2})$,
\item $\Omega := \Omega_0\otimes \Omega_0$.
\end{itemize}
As the net $\FerC$ is fermionic by nature, the interpretation of the scattering theory
of \cite{Buchholz75} is not clear. Nevertheless, we can show the following by an almost
same proof as in \cite[Lemma 5.2, Theorem 5.3]{Tanimoto11-3}.
\begin{lemma}\label{lm:net-fermionic}
The triple $(\M_\f,T,\Omega)$ is a Borchers triple.
\end{lemma}
\begin{proof}
To apply Proposition \ref{pr:general-construction}, it is immediate
that $S_\f$ commutes with translation since it is defined through
the spectral measure as above. It preserves $\H_{\FerC}\otimes \Omega_0$ and
$\Omega_0\otimes\H_{\FerC}$ pointwise, since these subspaces correspond
to the case where $m$ or $n$ is $0$ in the above decomposition
and $S_\f$ acts as the identity operator by definition. What remains to
show is the commutation property.

As we saw above, the operator $S_\f$ can be written as
\[
S_\f = \bigoplus_n \int \prod_j \Lambda(\widetilde{\f}(p_jP_1,\i_j)) \otimes \dE_1(p_1,\i_1)\otimes\cdots\otimes \dE_1(p_n,\i_n)\,.
\]
The point is that the operators which appear in the integrand
implement Longo-Witten endomorphisms as we saw above since
$p_j \ge 0$ in the support of the integration.

Let $x'\in \FerC(\RR_+)$ and consider $x'\otimes\1$ as an operator on
$\H_{\FerC}\otimes\H_{\FerC}^\sc$. We have
\[
\Ad S_\f (x'\otimes \1) =
\bigoplus_n \int \Ad \left(\prod_j \Lambda(\widetilde{\f}(p_jP_1,\i_j))\right)(x')
\otimes \dE_1(p_1,\i_1)\otimes\cdots\otimes \dE_1(p_n,\i_n)\,.
\]
Although this formula is not closed on $\H_{\FerC}\otimes\H_{\FerC}$,
the left hand side obviously restricts there. One sees that
the integrand remains in $\FerC(\RR_+)$.

Recall the operator $Z_0$ which gives the graded locality of $\FerC$.
One has to remind that $Z_0 = \frac{\1-\ima\G_0}{1-\ima}$ where $\G_0 = \Lambda(-\1)$,
hence $Z_0$ commutes with any second quantization operator. Then by the disintegration above
(and the corresponding disintegration with respect to the left component),
it is easy to see that $Z_0\otimes\1$ commutes with $S_\f$. 

Let us check the commutation property of the assumptions in Proposition \ref{pr:general-construction}.
Note that $\Ad Z_0(x)\otimes \1$ and $\Ad Z_0(x) \in \FerC(\RR_+)'$ for $x \in \FerC(\RR_-)$.
Since $Z_0\otimes\1$ and $S_\f$ commute as we saw above, to prove the first commutation relation,
it is enough to show that $[\Ad Z_0(x)\otimes\1, \Ad S_\f (x'\otimes\1)] = 0$ for
$x\in\FerC(\RR_-)$ and $x'\in\FerC(\RR_+)$. As operators acting on
$\H_{\FerC}\otimes\H_{\FerC}^\sc$, this is done by the above
disintegration of $\Ad S_\f(x'\otimes \1)$. Then both operators
naturally restrict to $\H_{\FerC}\otimes\H_{\FerC}$, and we obtain the claim
(cf.\! \cite[Lemma 5.2, Theorem 5.3]{Tanimoto11-3}).
The second commutation relation for Proposition \ref{pr:general-construction}
can be proven analogously.
\end{proof}
Finally we arrive at a new family of interacting Borchers triples with 
asymptotic algebra $\u1net\otimes \u1net$.
\begin{theorem}\label{th:net-bosonic}
Let us define
\begin{itemize}
\item $\N_\f := \{x\otimes \1, \Ad S_\f(\1\otimes y): x\in\FerC^{\uone}(\RR_-), y\in\FerC^{\uone}(\RR_+)\}''$,
\item $T(t,x) := T_0(\frac{t-x}{\sqrt 2})\otimes T_0(\frac{t+x}{\sqrt 2})$,
\item $\Omega := \Omega_0\otimes \Omega_0$.
\end{itemize}
Then the triple $(\N_\f,T,\Omega)$, restricted to $\overline{\N_\f\Omega}$,
is an asymptotically complete, interacting Borchers triple with
the asymptotic algebra $\u1net\otimes\u1net$ and scattering operator $S_\f|_{\overline{\N_\f\Omega}}$.
It also holds that
$\overline{\N_\f\Omega} = \overline{\u1net(I_+)\otimes\u1net(I_-)\Omega}$
for arbitrary intervals $I_+,I_-$.
\end{theorem}
\begin{proof}
Substantial arguments are already done: In Lemma \ref{lm:net-fermionic}
we constructed Borchers triples with $\FerC\otimes\FerC$ as the
asymptotic algebra. We have seen in Section \ref{u1-current-subnet}
the $\uone$-current net $\u1net$ is the fixed point subnet
of $\FerC$ with respect to the action of $\uone$. From the construction
in Section \ref{scattering-operator} and Theorem \ref{th:longo-witten-unitaries},
it is easy to see that $S_\f$ commutes with the product action of
the inner symmetries. Then all the statements of the Theorem
follow from the general consideration of
Proposition \ref{pr:asymptotic-completeness-fixed-point}.
\end{proof}

\subsection{Action of the S-matrix on the 1+1 particle space}\label{particle-production}
In this Section we want to analyze the action of the S-matrix of the models
constructed in Section \ref{scattering-operator} on the 
1+1 particle space $\Hbose^1 \otimes \Hbose^1$, i.e. one left and one right moving particle,
where we use the word particle in the sense of Fock space excitations.
We note that on the $n$+0 and 0+$n$ particle spaces $\H_n\otimes \CC \Omega_0$
and $\CC\Omega_0 \otimes\H_n$, respectively, the S-matrix $S$ acts trivially. 
A typical vector in $\Hbose^1 \otimes \Hbose^1$ is of the form 
$\Psi:=J(f)\Omega_0 \otimes J(g)\Omega_0$ which we express as the function 
$\Psi(p,\bar p) =\hat f(p)\hat g(\bar p)$.
The embedding  $\iota:L^2(\RR_+, p\dd p)\otimes L^2(\RR_+,\bar p \dd \bar p )
\cong \Hbose^1\otimes \Hbose^1 \hookrightarrow \H_{1,1}\otimes \H_{1,1} \subset \HferC \otimes \HferC$ is 
given by $\iota(\Psi)_{1,1;1,1}(p,q,\bar p,\bar q) = \frac1{(2\pi)^2}\Psi(p+q,\bar p + \bar q)$. 
We have an analogue of Lemma \ref{lem:LW1p}
\begin{proposition}\label{pr:production}
    Let $\f$ be some inner function. 
    The unitary $S_\f$ satisfies  $S_\f(\Hbose^1\otimes\Hbose^1) \subset \Hbose^1\otimes \Hbose^1$ 
    if and only if $\f(p) = \ee^{\ima(kp+\theta)}$.
\end{proposition}
\begin{proof}
    The action of $S_\f$ on $\Psi \in \H_1\otimes \H_1$ is given by 
    \[
        S_\f\Psi(p+q,\bar p+ \bar q) = \f(p\cdot \bar p) \check\f(q \cdot \bar p)
        \check\f(p \cdot \bar q) \f(q\cdot \bar q) \Psi(p+q,\bar p+ \bar q)
    \]
    which is again in $\Hbose^1\otimes \Hbose^1$ if it can be written as a function 
    $\tilde \Psi(p+q,\bar p + \bar q)$, in particular 
    if $\f(p\cdot \bar p) \check\f(q \cdot \bar p)
        \check\f(p \cdot \bar q) \f(q\cdot \bar q) = \widetilde \f(p+q,\bar p + \bar q)$. Setting $\bar p = 1$ and $\bar q = 0$, we have
        $\f(p) \check\f(q) = \widetilde\f(p+q,1)$.
        The rest follows as Lemma \ref{lem:LW1p}.

\end{proof}
\begin{remark} In the case $\varphi(p)=\ee^{\ima\kappa p}$,
one gets the models obtained in \cite{DT11-1} using warped convolution.
\end{remark}

\begin{proposition} Let $e$ be the projection on $\Hbose^1\otimes \Hbose^1$, then 
    $eS_\f e = \tilde\f(P\otimes P)$, where $\tilde\f$ 
    is boundary value of an analytic function in $\mathbb H$ with $|\tilde\f(p)|\leq 1$ and $P$ is the generator of translation restricted to
    the one-particle space (which gives rise the irreducible standard pair).
\end{proposition}
\begin{proof}
    It can be checked that
    \[
        (e_0f)(p,q)=\frac1{p+q} \int_0^{p+q} f(p+q-x,x)\dd x
    \]
    is the projection on $\Hbose^1 \subset \H_{1,1}$. Then the action of $eS_\f$ on 
    a $f \in \Hbose^1\otimes\Hbose^1$ can be calculated to be  $\f'(P\otimes 1, 1\otimes P)$
    with 
    \[
        \f'(p,q) =\frac1 {p\cdot q} \int_0^p\int_0^q\f( (p-x)\cdot(q-y) )\f( x\cdot y)\check\f( (p-x)\cdot y )\check\f( x\cdot (q-y)) \dd y\dd x
    \]
    and it is easy to check that with $\tilde\f(p):=\f'(p,1)$ it holds $\f'(p,q)=\tilde\f(p\cdot q)$
    for all $p,q>0$. That $|\tilde \f(p)|\leq 1$ can be  checked directly or follows from the fact that
    $S_\f$ is unitary.
\end{proof}
\begin{remark}
It is a general feature of asymptotically complete Borchers triples with
asymptotic algebra $\u1net\otimes\u1net$ that the restriction of the
scattering matrix $S$ to $eSe$ is
a functional calculus of $P\otimes P$.
Indeed, both $e$ and $S$ commute with the translation $T$,
but $T$ is maximally abelian when restricted to $\Hbose^1\otimes\Hbose^1$,
hence there is a function $\f_S$
such that $eSe = \f_S(P\otimes \1, \1\otimes P)$. Furthermore,
both $e$ and $S$ commute with boosts, so does $\f_S$ and
one obtains the form $eSe = \f_S'(P\otimes P)$.
\end{remark}

We note that the proof above shows that $|\tilde\f(M^2/2)|$ is the probability 
that an improper state in $\H^1_\u1net\otimes \H^1_\u1net$ 
with mass $M^2$ is scattered elastically in 
the sense of Fock space particles, where
\[
    \tilde\f(p)=\frac1p\int_0^p\int_0^1\f( (p-x)(1-y) )\f( xy)\check\f( (p-x)y )\check\f( x(1-y)) \dd y\dd x
    \,.
\]

As we discussed in \ref{u1-current}, the Hilbert space of the
$\uone$-current net, and hence the tensor product of two copies of it,
admit the bosonic Fock space structure,
hence we can consider the particle number. Although we admit that
this concept does not have an intrinsic meaning, we claim that
it is possible to interpret this as the number of massless particles.

An evidence comes from the comparison with massive cases.
In \cite{Lechner08} Lechner has constructed a family of massive interacting
models parametrized by so-called scattering functions, and later he
reinterpreted them as deformations of the massive free field \cite{Lechner11}.
If one applies the same deformation procedure to the derivative of
the massless free field whose net is $\u1net\otimes\u1net$
(with scattering functions satisfying $S_2(0) = 1$),
he obtains the Borchers triples with $\u1net\otimes\u1net$ as the
asymptotic net constructed in \cite{Tanimoto11-3} \footnote{Private communication
with Gandalf Lechner and Jan Schlemmer. This will be presented elsewhere.}.
Hence the models in \cite{Tanimoto11-3} should be considered as the massless
versions of the models in \cite{Lechner08}. Likewise, it can be said that
the models constructed in the present paper are the deformed (in an appropriate sense)
version of the massless free field.

In massive case, there is a mass gap in the spectrum of the spacetime translation
and the one-particle space of the Fock space has an intrinsic meaning. In massless case,
such an intrinsic interpretation is lost but there is still the Fock space structure.
Thus we think that, if the two-particle space in the Fock structure is not preserved by
the S-matrix, as in the case where $\f$ is not exponential (see Proposition \ref{pr:production}),
then it represents massless particle production.

\section{Conclusion and outlook}\label{conclusion}
In this paper we have constructed a new family of Longo-Witten
endomorphisms on $\u1net$ through the inclusion
$\u1net = \FerC^\uone \subset \FerC$. We combined them to
construct interacting
wedge-local nets with $\u1net\otimes\u1net$ as the asymptotic algebra and 
showed that their S-matrices
do not preserve the $n$-particle space of the bosonic Fock space.
Particle production is a necessary feature of interacting models
in higher dimensions \cite{Aks65}, thus this result opens up
some hope for algebraic construction of higher dimensional
interacting models.

However, there are at least two shortcomings with the present
method. The first is that we proved only wedge-locality of the models.
As already shown in \cite{Tanimoto11-3}, a wedge-local net can
be dilation-covariant and at the same time interacting. On the other hand,
a {\em strictly local} dilation-covariant (asymptotically complete) net
is necessarily not interacting \cite{Tanimoto11-2}. Hence, interaction
of wedge-local nets could be just a false-positive and strict locality
is desired. The second is the fact that the concept of particle
in massless case is not intrinsically defined. Although the Fock space
structure is easily understood, its interpretations should be treated
with care.

These issues could be overcome by considering massive cases.
As for strict locality, it has been shown that the deformation of
the massive free field by a suitably regular function is
again strictly local \cite{Lechner08, Lechner11}. On the other hand,
in massless situation, even the simplest case $\f(p) = -1$ (where
$\f$ is an inner symmetric function used in \cite{Tanimoto11-3}
to deform {\em directly} $\u1net\otimes\u1net$) is already
not strictly local \cite{Tanimoto11-3}. Hence we believe that strict
locality should be addressed in massive models. Furthermore,
for a massive asymptotically complete model, the notion of
particle production is intrinsic.
Fortunately, it is known that the construction in \cite{Tanimoto11-3}
coincides with the deformation of the massive free field
as we remarked in the last section,
hence a further correspondence between massive and massless
cases are expected.
We hope to investigate this problem in a future publication.

Of course, interacting models in higher dimensions are always
one of the most important issues. Although conformal nets themselves
are not interacting \cite{BF77}, some new constructions based
on CFT could be possible and ideas from the present article
could be useful.

\subsubsection*{Acknowledgment.}
We thank our supervisor Roberto Longo for his constant support and
useful suggestions.
Y.\! T.\! thanks Gandalf Lechner and Jan Schlemmer for discussions
on the relation between the present construction and the deformation
of \cite{Lechner11}.

\def\cprime{$'$}

\end{document}